\documentclass[conference]{IEEEtran}
\IEEEoverridecommandlockouts
\usepackage{cite}
\usepackage{amsmath,amssymb,amsfonts,amsthm,latexsym}
\usepackage{algorithmic}
\usepackage{graphicx}
\usepackage{textcomp}
\usepackage{xcolor}

\usepackage{setspace}
\usepackage{tabularx,ragged2e,booktabs,caption}
\usepackage{thmtools}
\usepackage{thm-restate}
\usepackage{hyperref}
\usepackage{cleveref}

\usepackage{comment}
\usepackage[normalem]{ulem}

\usepackage{algorithm}

\usepackage{paralist}
\usepackage{multicol}
\usepackage{hyperref}
\usepackage{comment}

\usepackage{lineno}
\newtheorem{observation}{Observation}

\newtheorem{theorem}{Theorem}
\newtheorem{lemma}{Lemma}

\newcommand{\sapta}[1]{{\color{black} #1}}
\newcommand{\remove}[1]{}
\newcommand{\AlgName}{{\sc ABCC}}
\newcommand{\knows}[3]{$\mbox{\it SysInfo}^{{#3}} \subseteq \mbox{\it Server\_Changes}_{{#1}}^{{#2}}$}

\def\nats{{\mathbb N}}

\def\BibTeX{{\rm B\kern-.05em{\sc i\kern-.025em b}\kern-.08em
    T\kern-.1667em\lower.7ex\hbox{E}\kern-.125emX}}
\begin{document}

\title{Byzantine-Tolerant Register in a System with Continuous Churn\\
}

\author{\IEEEauthorblockN{ Saptaparni Kumar}
\IEEEauthorblockA{\textit{Department of Computer Science} \\
\textit{Boston College}\\
Chestnut Hill, USA \\
saptaparni.kumar@bc.edu}
\and
\IEEEauthorblockN{ Jennifer L. Welch}
\IEEEauthorblockA{\textit{Department of Computer Science and Engineering} \\
\textit{Texas A\&M University}\\
College Station, USA \\
welch@cse.tamu.edu}
}

\maketitle

\begin{abstract}

A shared read/write register emulation provides the illusion of shared-memory on top of message-passing models. 
The main hurdle with such emulations is dealing with server faults in the system. 
Several crash-tolerant register emulations in static systems require algorithms to replicate the value of the shared register onto a majority of servers. Majority correctness is necessary for such emulations. Byzantine faults are considered to be the worst kind of faults that can happen in any distributed system. Emulating a Byzantine-tolerant register requires replicating the register value on to more than two-thirds of the servers. Emulating a register in a dynamic system where servers and clients can enter and leave the system and be faulty is harder than in static systems. 
There are several crash-tolerant register emulations for dynamic systems. 

This paper presents the first emulation of a multi-reader multi-writer atomic register 
in a system that can withstand nodes continually entering and leaving, imposes no upper bound on the system size and can tolerate Byzantine servers. 
The algorithm works as long as the number of servers entering and leaving during a fixed time interval is at most a constant fraction $\alpha$ of the system size at the beginning of the interval, and as long as the number of Byzantine servers in the system is at most $f$. Although our algorithm requires that there be a constant 
known upper bound on the number of Byzantine servers, this restriction 
is unavoidable, as we show 
that it is impossible to emulate an atomic register 
if the system size and maximum number of servers that can be Byzantine in the system is unknown to the nodes.

\end{abstract}

\section{Introduction}\label{section:introduction}

A long-standing vision in distributed systems is to build reliable systems from unreliable components. We are increasingly dependent on services provided by distributed systems resulting in added vulnerability when it comes to failures in computer systems.  In a dependable computing system, the term ``Byzantine" fault is used to represent the worst kind of failures imaginable. Malicious attacks, operator mistakes, software errors and conventional crash faults are all encompassed by the term Byzantine faults~\cite{LamportSP82}. The growing reliance of industry and government on distributed systems and increased connectivity to the Internet exposes systems to malicious attacks. Operator mistakes are also a very common cause of Byzantine faults~\cite{MurphyL00}. The growth in the size of software in general leads to an increased number of software errors. Naturally, over the past four decades, there has been a significant work on consensus and replication techniques that tolerate Byzantine faults~\cite{LamportSP82, CastroL02, CastroL99, AlvisiPMRW00} as it promises dependable systems that can tolerate any type of bad behavior.
The shared-memory model is a more convenient programming model than message-passing, and shared register emulations provide the illusion of shared-memory on top of message-passing models.
	\begin{center}
		\fbox{
			\parbox{220pt}{
				In this paper, we emulate the first Byzantine-tolerant atomic register on top of a dynamic message-passing system that never stops changing. 
			}
		}
	\end{center}
 \subsection{Related Work}
Typically,  crash-tolerant emulations~\cite{AttiyaBD1995, LynchS1997} of a shared read/write register replicate the value of the register
in multiple servers and require
readers and writers to communicate with majority of servers. For instance, the ABD emulation~\cite{AttiyaBD1995}
replicates the value of the shared register in a static set of servers.  It assumes that a minority of the servers may fail 
by crashing. 
This problem of emulating a shared register has been extended to static systems with servers subject to Byzantine faults and these  \sapta{emulations typically assume that} 
two-thirds~\cite{AttiyaB2003}  or three-fourths~\cite{AbrahamCKM2007} of the servers are non-faulty.  It is shown in ~\cite{MartinAD2002} that more than two-third correctness is necessary for Byzantine-tolerant register simulation.  
Byzantine quorum systems (BQS)~\cite{AlvisiPMRW00, MalkhiR98, MalkhiRW00, MartinAD2002} are a well known tool for ensuring consistency and availability of a shared register.  A BQS is a collection of subsets of servers, each pair of which intersect in a set containing sufficiently many correct servers to guarantee consistency of the replicated register as seen by clients. Vukolic~\cite{Vukolic2012} provides an extensive overview of the evolution of quorum systems for distributed storage and consensus. Dantas et. al~\cite{DantasBFC2007} present a comparative evaluation of several Byzantine quorum based storage algorithms in the literature.

The success of this replicated approach for static systems, where the set of readers, writers, and servers
is fixed, has motivated several similar emulations for {\em dynamic} systems, where nodes may enter
and leave. Change in system composition due to nodes entering and leaving is called {\em churn}. Ko et. al~\cite{KoHG2008} provide a detailed discussion of churn behavior in practice.
Most existing emulations of atomic registers in dynamic systems deal with crash-faults and rely either on the assumption that churn eventually stops for a long enough period
(e.g., DynaStore~\cite{AguileraKMS2011} and RAMBO~\cite{LynchS2002})
or on the assumption that the system size is bounded (e.g.,~\cite{BaldoniBKR2009},~\cite{BaldoniBR12}). 
Attiya et al.~\cite{AttiyaCEKW2015, AttiyaCEKW2018} proposed an emulation of a crash-tolerant shared register in a system that does not require churn to ever stop. 
Bonomi et al.~\cite{BonomiPPT16} present an emulation of a server based regular read/write storage in a synchronous message-passing system that is subject to ``mobile Byzantine failures". They prove that the problem is impossible to solve in an asynchronous setting. 
The system size, however, is fixed and mobility, in this paper refers to the Byzantine agents that can be moved from server to server. 
Baldoni et al.~\cite{BaldoniBN13} provide the first emulation of a Byzantine-tolerant {\em safe}~\cite{Lamport86a} register in an eventually synchronous system with churn but the size of the system is upper bounded by a known parameter. To the best of our knowledge, there is \sapta{no work} on implementing an \sapta{atomic register} in a (i) dynamic system (ii) with no upper bound on the system size  where (iii) servers are subject to Byzantine faults and (iv) any number of clients can crash.

\subsection{Contributions}
The first contribution \sapta{of this paper is the first 
algorithm to emulate an} 
atomic multi-reader/multi-writer register that does not require churn to ever stop, does not have an upper bound on the system size and tolerates up to a constant number of Byzantine servers in the system. \sapta{It is a common practice to assume that clients cannot be Byzantine in a system~\cite{MartinAD02, BaldoniBN13} as a Byzantine client can maliciously contact separate sets of servers and write different values which
results in an inconsistent register thus violating safety. In our model, clients can only crash.} 

Although our algorithm requires that there be a constant known upper bound on the number of Byzantine servers that can be tolerated, this restriction 
is unavoidable as shown in our second contribution. Our second contribution a proof that it is impossible 
to emulate an atomic register in a system with churn if the maximum number of Byzantine servers is unknown to the nodes. 
 

Our system model is similar to the one in~\cite{AttiyaCEKW2018}. We 
assume that there exists a parameter $D$, an upper bound, unknown to the nodes, on the delay of any message (between correct nodes). There is no lower bound on message delays and nodes do not have
real-time clocks.  

Churn is modeled as follows: we assume that in any time interval of length $D$, the number of servers that enter or leave the system is at most a constant fraction, $\alpha$ (known to all nodes), of the number of
servers in the system at the beginning of the interval. 
We also assume the messages are authenticated with {\em digital signatures}~\cite{LindellLR06}. In real world systems digital signatures in messages are implemented using public-key signatures~\cite{RivestSA1978} and  message authentication codes~\cite{Tsudik1992}. Intuitively, this means that Byzantine servers cannot lie about the sender of a message.
\subsection{Challenges with Byzantine Servers}
Our algorithm is called \sapta{\AlgName{} register, for \emph{Atomic Byzantine-tolerant  Continuous Churn}.} 
Unlike crash faults, data may easily be corrupted by Byzantine servers by sending old information, modified information or even  different information to different sets of nodes while replying to a particular message. A Byzantine server may choose to not reply to a message at all, even if it is active or it may even choose to reply to a single message multiple times. 
Our algorithm uses a mechanism where at every stage of the algorithm, nodes wait for at least $f+1$ replies from distinct servers before taking any major steps, to make sure at least one reply from a non-faulty server was received. \sapta{The algorithm also has to make sure that a decision is not affected by multiple replies corresponding to a single message from Byzantine servers. 
 A Byzantine server can also ``lie" about its joined state and propagate misinformation by sending out messages after ``pretending" to leave the system. } 





\section{System Model and Problem Statement} \label{section:model}

We model each node $p$, whether client or server, as a state machine with a
set of states.
The set of states for $p$ contains two initial states, $s_p^i$ and $s_p^\ell$.
Initial state $s_p^i$ is used if node $p$ is in the system initially, whereas
$s_p^{\ell}$ is used if $p$ enters the system later.

State transitions are triggered by the occurrences of events.
Possible triggering events include node $p$ entering the system 
({\sc Enter$_p$}),
leaving the system ({\sc Leave}$_p$), and receiving a message $m$ 
({\sc Receive}$_p(m)$).
In addition, triggering events for a client $p$ include
the invocation of an operation ({\sc Read}$_p$ or {\sc Write}$_p(v)$) 
and the client crashing ({\sc Crash}$_p$).

A {\em step} of a node $p$ is a 5-tuple $(s',T,m,R,s)$, where $s'$ is the
old state, $T$ is the triggering event, $m$ is the message to be sent, $R$
is either a response ({\sc Return}$_p(v)$, {\sc Ack}$_p$, or {\sc Joined}$_p$)
or $\bot$, and $s$ is the new state.  
The message $m$ includes an indication as to whether it should be sent to
all servers or to all clients, indicated as ``s-bcast'' or ``c-bcast''.
{\sc Return}$_p$ is the response to {\sc Read}$_p$, {\sc Ack}$_p$ 
is the response to {\sc Write}$_p$,
and {\sc Joined}$_p$ is the response to {\sc Enter}$_p$; these responses are
only done by clients.
If $T$ is {\sc Crash}$_p$, then $m$ and $R$ are both $\bot$.

If the values of $m$, $R$, and $s$ are determined by
the node's transition function applied to $s'$ and $T$, then the step is said
to be {\em valid}. A step by a client is always valid. 
In an invalid step (taken by server $p$), the values of $m$, $R$, and $s$ can be arbitrary, with the restriction that $p$ cannot modify values containing information about node ids. There is more detail on how this assumptions applies to our algorithm in Section~\ref{section:algorithm}.

A {\em view} of a node $p$ is a sequence of steps such that:
\begin{itemize}
\item the old state of the first step is an initial state;
\item the new state of each step equals the old state of the next step
\item if the old state of the first step is $s_p^i$, then no {\sc Enter}$_p$
      event occurs
\item if the old state of the first step is $s_p^{\ell}$, then the triggering
      event of the first step is {\sc Enter}$_p$ and there is no other 
      occurrence of {\sc Enter}$_p$; 
\item at most one {\sc Leave}$_p$ occurs, and if it occurs there are no 
      later steps;
\item at most one {\sc Crash}$_p$ occurs;
      if it occurs, then $p$ is a client and there are no later steps.
\end{itemize}

A view is {\em valid} if every step in it is valid.

In our model, a node that leaves the system cannot re-enter and a client node
that crashes cannot recover.

{\em Time} is represented by nonnegative real numbers.  A {\em timed view}
is a view whose steps are labeled with nondecreasing times (the real times
when the steps occur) such that:
\begin{itemize}
\item for each node $p$, if the old state in $p$'s first step is $s_p^i$,
      then the time of $p$'s first step is 0 (such a node is said to be
      {\em in the system initially}), and
\item  if a view is infinite, the step times must increase without bound.  
\end{itemize}

Given a timed view of a node $p$, if $(s',T,m,R,s)$ is the
step with the largest time less than or equal to $t$, then $s$ is the 
{\em state of $p$ at time $t$}.
A node is said to be {\em present} at time $t$ if its first step has
time at most $t$ but has not left (i.e., {\sc Leave}$_p$ 
does not occur at or before time $t$).
The number of servers present at time $t$ is denoted by $NS(t)$.
A node is said to be {\em active} at time $t$ if it is present at $t$
and {\sc Crash}$_p$ has not occurred before time $t$. Since servers never 
experience crashes, a server that is present is also active.

We define the following system parameters that are used in the upcoming
definition of an execution:
\begin{itemize} 

\item $D > 0$ is the maximum message delay (the {\em delay}
  of a message sent at time $t$ and received at time $t'$ is $t'-t$).
\item $\alpha > 0$ is the {\em churn rate}, which bounds how fast servers
  can enter and leave; there are no bounds for clients.
\item $f \ge 1$ is the maximum number of Byzantine-faulty servers.
\item $NS_{min} > 0$ is the minimum number of servers. This value is unknown to all nodes in the system.
\end{itemize}
The parameters $\alpha$ and $f$ are known to the nodes, but $D$ is not.

An {\em execution} is a possibly infinite set of timed views, one for each
node that is ever present in the system, that satisfies the following
assumptions:
\begin{enumerate}
\item[A1:] For all $t \ge 0$, the number of nodes present at time $t$
    is finite and $NS(t) \ge NS_{min}$.
\item[A2:] Every message s-bcast (respectively, c-bcast) has at most one 
   matching receipt at each server (respectively, client)
  and every message receipt has exactly one matching bcast.
\item[A3:] If message $m$ is s-bcast (respectively, c-bcast) at time $t$ 
  and server (respectively, client) node $q$ is active
  throughout $[t,t+D]$, then $q$ receives $m$.   
  The delay of every received message is in $(0,D]$.
\item[A4:] Messages from the same sender are received in the order they
  are sent (i.e., if node $p$ sends message $m_1$ before sending message $m_2$,
  then no node receives $m_2$ before it receives $m_1$).  
\item[A5:] For all times $t > 0$, the number of {\sc Enter} and {\sc Leave} 
  events for servers in $[t,t+D]$ is at most $\alpha \cdot NS(t)$.  
\item[A6:] The timed view of every client is valid.  A client $p$
  whose timed view does not contain {\sc Crash}$_p$ is a {\em correct}
  client.  There is a set $F$ of servers, with $|F| \le f$, such that
  the timed view of every server not in $F$ is valid.  The servers not
  in $F$ are {\em correct} servers.  The servers in $F$ are 
  Byzantine-faulty.
\item[A7:] If a {\sc Read}$_p$ or {\sc Write}$_p$ invocation occurs at
  time $t$, then $p$ is an active client that has already joined ({\sc
  Join}$_p$ occurs before $t$).
  (no {\sc Leave}$_p$ or {\sc Crash}$_p$ occurs by time $t$).
\item[A8:] At each client node $p$, no {\sc Read}$_p$ or {\sc Write}$_p$ 
  occurs until there have been responses to all previous
  {\sc Read}$_p$ and {\sc Write}$_p$ invocations.
\end{enumerate}

Assumption A1 states that the system size is always finite and there
is always some minimum number of servers in the system.  Assumptions
A2 through A4 model a reliable broadcast communication service that
provides nodes with a mechanism to send the same message to all
servers or to all clients in the system with message delays in $(0,D]$. 
However, Byzantine servers
may choose not to s-bcast (respectively, c-bcast) the same message to
every server/client in the system and just send different unicast
messages to different nodes in the system.  Assumption A5 bounds the
server churn.  Assumption A6 bounds the number of Byzantine-faulty
servers and restricts the clients to experience only crash failures.  
Assumptions A7 and A8 ensure that operations are only
invoked by joined and active clients and at, any time, at most one
operation is pending at each client node.

We consider an algorithm to be {\em correct} if every execution of
the algorithm satisfies the following conditions : 
\begin{itemize}
\item[C1:]
  For every client $p$ that is in the system initially,
  {\sc Joined}$_p$ does not occur.
  Every client $p$ that enters the system later
  and does not leave or crash eventually joins.
\item[C2:] In the view of each client $p$, ignoring message-receipt events,
  each {\sc Read}$_p$ or {\sc Write}$_p$ is immediately followed by 
  either {\sc Leave}$_p$, {\sc Crash}$_p$, or a matching response
  ({\sc Return}$_p$ or {\sc Ack}$_p$), and each {\sc Return}$_p$ or 
  {\sc Ack}$_p$ is immediately preceded by a matching invocation 
  ({\sc Read}$_p$ or {\sc Write}$_p$).
\item[C3:] The read and write operations are {\em atomic} (also called
  {\em linearizable}) \cite{Lamport86, Lamport86a, HerlihyW1990}:
  there is an ordering of all completed reads and writes and some subset
  of the uncompleted writes such that every read returns the value of the
  latest preceding write (or the initial value of the register if there is
  no preceding write) and, if an operation $op_1$ finishes before another
  operation $op_2$ begins, then $op_1$ occurs before $op_2$ in the ordering.
\end{itemize}
It is the responsibility of the algorithm to complete joins, reads and
writes, and choose the right values for the reads, as long as Assumptions
A1--A8 are satisfied.

Although our model places an upper bound on message delays, it does
not place any lower bound on the message delays or on local
computation times.  Moreover, nodes cannot access clocks to measure
the passage of real time.  Consequently, the well-known consensus
problem is unsolvable in our model as proved in
~\cite{AttiyaCEKW2018}, just as it is unsolvable in a model with no
upper bound on message delays \cite{FLP}.


\section{Impossibility of a Uniform Algorithm with Byzantine Servers}
\label{section:impossibility}

\sapta{
In this section we show it is impossible to emulate an atomic register in our system model if nodes have no information about the maximum number of Byzantine-faulty servers or about the total number of servers.  We model the lack of such information through the notion of a ``uniform" algorithm. There have been similar results proved for the impossibility of self-stabilization~\cite{AngluinAFJ08} and consensus~\cite{AlchieriBGF18} in Byzantine-tolerant systems with unknown participants.  } 

An algorithm is called {\em uniform} if the code run by every node is
independent of both the system size and the maximum number of
Byzantine servers in the system. Thus in a uniform algorithm, 
 \sapta{for any 
particular node id $p$, there are only two
    possible state machines that $p$ can have.  One is for the situation 
when $p$ is in the system initially
    and the other is for the situation when $p$ enters the system later.  
Otherwise, the state machines
    are completely independent of the initial size of the system (or the 
size when $p$ enters) and of the
    maximum number of Byzantine servers. }


\begin{theorem}
It is impossible to emulate an atomic\footnote{This theorem and proof holds for a {\em safe}~\cite{Lamport86a} register as well}  read/write register in our
model with a uniform algorithm.
\end{theorem}

\begin{proof} 
Suppose in contradiction, there is a uniform algorithm, $\mathbb{A}$,
which simulates an atomic register and 
\sapta{can tolerate $f$ Byzantine server failures for any  $f \in \mathbb{N}^+$, as long as the system size at all times is at least $NS_{min}(f)$ 
for some function $NS_{min}$.}
Consider the following executions of $\mathbb{A}$.

\textbf{Execution} $e_1$: The maximum number of Byzantine servers is
$1$. The set of servers in the system initially is $S_1$, where
$|S_1|= NS_{min}(1)$, and all servers are correct. All message delays
are $D$. The initial value of the simulated register is $v$. A new
client $p$ enters the system at time $t_e$ and joins by time
$t_j\ge t_e$. Client $p$ invokes a read on the simulated register at
time $t_r> t_j$. No other operation on the simulated register is
invoked. By assumed correctness of algorithm $\mathbb{A}$, the read
invoked by $p$ returns $v$ at some time $t'_r \ge t_r$.

\textbf{Execution} $e'_1$: Multiply the real time of every event in
$e_1$ by $\min\{\frac{D}{t'_r},1\} $. As a result, all events in the
time interval $[0,t'_r]$ in $e_1$ are compressed into the interval
$[0,D]$ in $e'_1$.

\textbf{Execution} $e_2$: The maximum number of Byzantine servers is
$f_2 = |S_1|$.  The set of servers in the system initially is
$S_2$, where $|S_2| = NS_{min}(f_2)$and $S_1$ is a subset of $S_2$. 

Note that the system sizes 
are chosen such that we have execution $e_1$ with $|S_1|$ servers, all correct, and
    execution $e_2$ with exactly $|S_1|$ Byzantine servers.

 
There is at least one client  in the system initially.  All
message delays are $D$.  The initial value of the simulated register
is $v$.  No churn happens in this execution.  Client $q$ that was
in the system at time 0 invokes a write of $v' \neq v$ at time
$t_w$. By the assumed correctness of algorithm $\mathbb{A}$, the write
completes at some time $t'_w\ge t_w$. No other operation on the
simulated register is invoked.
		 
Finally we construct a prefix $e_3$ of a new execution from
executions $e_1'$ and $e_2$.  First, we specify a set of timed 
views 
and then we show that this set indeed forms the prefix of an
execution.  Let $e_3^1$ be the set of timed views in $e_2$.  Note that
$e_3^1$ includes the write operation invoked by client $q$.
Truncate each timed view in $e_3^1$ immediately after the latest
step with associated time at most $t_w'$, i.e., just after the write
by client $q$ finishes.  Then append steps that result in the
immediate delivery of all messages that are in transit at $t_w'$. Call
the resulting set of timed views $e_3^2$.  Construct $e_3$ from
$e_3^2$ as follows.

\textbf{Execution prefix} $e_3$: Add to the set the prefix of the timed view
of client $p$ from $e_1'$ that ends at time $D$, but change the time
associated with each step by adding $t_w'$ to it.  For each server $s$
in $S_1$, append the prefix of $s$'s timed view in $e_1'$ that ends at
time $D$, but change the time associated with each step by adding $t_w'$ to
it.  Append nothing to the timed views for the remaining nodes (client
or server).

The idea behind $e_3$ is to have all the nodes behave correctly
through $q$'s write of $v'$, and then have a new client $p$ enter,
join, and invoke a read during which time it communicates only with
the servers in $S_1$.  However, the servers in $S_1$ are
Byzantine and start acting as they did in $e_1'$, causing $p$'s read to
incorrectly return the value $v$, instead of $v'$.  An important
technicality in the construction of $e_3$ is to adjust the time of
steps taken from $e_1'$. \sapta{The assumed uniform nature of the algorithm is what allows us to combine timed views from $e_1$, in which at most one server can be Byzantine,  with timed views from $e_2$, in which $f_2 > 1$ servers can be Byzantine.}

In order for the existence of the incorrect read by $p$ in $e_3$ to
contradict the assumed correctness of $\mathbb{A}$, we must show that
$e_3$ is the prefix of an execution (otherwise, bad behavior by
$\mathbb{A}$ is irrelevant).

We show that $e_3$ is a prefix of an execution by verifying properties
$A_1$ through $A_8$.  $A_1$, $A_5$, and $A_6$--$A_8$ are clear.  


$A_2$--$A_4$: Every message sent by a node (client or server) has exactly one
matching receipt. We show this in two parts: (i) If the message was
sent before $t'_w$, it was either delivered before $t'_w$ (from $e_2$)
or at $t'_w$ if it was pending at $t'_w$ (from construction of $e_3$).
(ii) If the message was sent after $t'_w$: Messages exchanged between
$p$ and $S_1$ after $t'_p$ are all delivered within $D$ time (from
$e'_1$) and all other messages in $e_3$ after $t'_w$ are delivered
with delay $D$.
All message delays in $e_1'$ are $\le D$ and the message delays
in $e_2$ are $D$. Therefore, the message delays in $e_3$ are at most $D$.  





In $e_3$, $p$'s read returns $v$, whereas the latest preceding write
wrote $v' \neq v$. The value returned by node $p$ is incorrect and as
$e_3$ is the prefix of an execution, this violates the safety property
of the register. Therefore, it is impossible to simulate a shared
register in dynamic systems where new nodes entering have no
information about the system size and no information on the maximum number of
Byzantine servers present in the system.
\end{proof}

\makeatletter
\newcommand{\setalglineno}[1]{%
  \setcounter{ALC@line}{\numexpr#1-1}}
\makeatother

\section{The \AlgName{} Algorithm}
\label{section:algorithm}
  
  The \AlgName{} algorithm is loosely based on the algorithm in~\cite{AttiyaCEKW2018} along with modifications to accommodate the new client-server model\sapta{, as opposed to the peer-to-peer model, 
  and Byzantine servers as opposed to crashes in~\cite{AttiyaCEKW2018}. The \AlgName{} algorithm} is divided into two main parts: Algorithm~\ref{algo:Server} for servers and Algorithm~\ref{algo:Client} for clients. Algorithm~\ref{algo:Common} contains a set of common procedures used by both servers and clients. The server algorithm contains a mechanism for tracking the composition of
the system with respect to servers and for assisting clients with reads and writes. The client algorithm is for newly entered clients to join the system and for joined clients to read from and write to the shared register. 

\sapta{Initially the system consists of a set of servers $S_0$ and a set of clients $C_0$ such that $|S_0| \ge NS_{min}$ and $|C_0| \ge 0$. A server $p\in S_0$ is joined at time $0$ and it knows about all other servers $q\in S_0$. In the following code description, we use the convention that local variables of node $p$ are subscripted with $p$.}  Each  node $p$ maintains a set of events, {\it Server\_Changes$_p$}, concerning the servers that have entered, joined and left the system. 
A node $p$ also maintains the set
$\mbox{\it Present}_p $ 
 that stores information about servers that have entered, but have not left, as far as $p$ knows.  A server $p$ is called a {\em member} if it has joined the system but not left. 
Client  $p$ maintains the derived variable $\mbox{\it Members}_p$ 
of servers that $p$ considers as members.  
The variables $val_p$, $num_p$ and $w\_id_p$ 
 store the latest register value and its timestamp, consisting of the ordered pair $(num_p, w\_id_p)$, known by $p$. 
  \sapta{The variable $w\_id_p$ stores the id of the writer (client) that wrote $val_p$}. The set {\it Known\_Writes}$[\cdot]_p$ stores an entry for all nodes $q$ that $p$ thinks have entered. An entry {\it Known\_Writes}$[q]_p$ stores all  the values written to the register that server $q$ has declared to know about.  
 So, at all times {\it Known\_Writes}$[p]_p$ stores the values of all writes (i) $p$ has heard of through an ``update" message from a client performing a write or (ii) that occur in more than $f$ entries in {\it Known\_Writes}$[\cdot]_p$.

 Algorithm~\ref{algo:Server}:  When a server $p$ enters the system, it 
  broadcasts to all the servers 
  an enter message requesting information about prior events. 
When a server $q$ finds out that node $p$ has entered (or joined or left) the system, 
$q$ updates {\it Server\_Changes$_q$} accordingly and  
 sends out an echo message with information about the system (stored in $\mbox{\it Server\_Changes}_q$) and the shared register (stored in the variable {\it Known\_Writes}$[q]_q$). 
 When node $p$ receives  at least $f+1$  enter-echo messages from  joined servers (to make sure at least one reply is from a correct server),  it calculates 
\sapta{
the number of replies it needs in 
order to join as a fraction of the
    number of servers it believes are present, i.e., $\gamma \cdot |Present_p|$.  
This value is stored in the
    $join\_bound_p$ local variable.} 
Setting $\gamma$ is a key
challenge in the algorithm as setting it too small might not propagate
updated information, whereas setting it too large might not guarantee
termination of the join. A server $q \notin S_0$ is considered joined once it has executed line number~\ref{line: server joined} of Algorithm~\ref{algo:Common}.

The algorithm sends out messages that are authenticated with digital signatures. As a result Byzantine servers can send out incorrect information about everything except for node ids.
Byzantine servers can modify information about anything sent out in 
messages of Algorithm~\ref{algo:Server} subject to the following restrictions: 
\begin{itemize}
\item   {\it Server\_Changes}$_p$: A Byzantine server $p$ can only send out subsets of the  {\it Server\_Changes}$_q$ set for some $q$ that has entered the system. Server $p$ cannot modify entries as each entry in this variable contains a server node id which was digitally signed by the sending server. 
\item $val_p$, $num_p$ and $w\_id_p$: A Byzantine server can modify variables $val_p$ and $num_p$ while sending them out. But it cannot modify  the $w\_id_p$ variable which is the id of the client that invoked the write or $\perp$.


\item {\it Known\_Writes}$[\cdot]_p$: A Byzantine server can send out subsets of {\it Known\_Writes}$[\cdot]_p$, but cannot add entries. For an entry 
$(val_q, (num_q, w\_id_q))\in $ {\it Known\_Writes}$[q]_p$, Byzantine servers can  modify the $val$ and $num$ variables of this entry for node $q$.   
\end{itemize}
The JoinProtocol$_p$ procedure in Algorithm~\ref{algo:Common} is used by both newly entered servers and clients to join the system. Once joined\footnote{Note that joining is not the same as entering. Once a node (client or server) enters the system, it has to complete running the JoinProtocol subroutine to become ``joined".}, servers can reply to read/write queries from clients. In addition to that, for all nodes $p$, there exists an in-built procedure, IsClient$_p$($q$) that can check, based on the 
node id $q$, if $q$ is a client or not. 
This procedure  
prevents Byzantine servers from pretending to be clients. 
\begin{algorithm*}[!hptb]
\begin{algorithmic}[1]
	\small
	\footnotesize
	\item[]{\bf In-built Procedure:}
	\item[] { IsClient($q$) }\COMMENT{returns $true$ if $q$ is a client and $false$ if $q$ is a server } 
	\item[] {\bf Local Variables:}
	\item[] $\mbox{\it Server\_Changes}$ \COMMENT{set that stores information about entering, leaving and joining of servers known by $p$ }	 		\item[] \hspace*{.5in}\COMMENT{initially $\{enter(q) ~|~ q \in S_0\} \cup \{ join(q) ~|~ q \in S_0\}$,  if $p$ is in the system at time 0  and $\emptyset$ otherwise }      			
	\item[] $join\_bound$ \COMMENT{if non-zero, the number of enter-echo messages $p$ should receive before joining; initially $0$}
	\item[] $enter\_echo\_counter$ \COMMENT{number of enter-echo messages received so far; initially $0$}
	\item[] $enter\_echo\_from\_joined\_counter$ \COMMENT{number of enter-echo messages from joined servers received so far; initially 0}
	\item[] $is\_joined$ \COMMENT{Boolean to check if $p$ has joined the system; initially $false$}
	 \item[] $val$ \COMMENT{latest register value known to $p$; initially $\perp$}
        \item[]  $num$ \COMMENT{sequence number of latest value known to $p$; initially 0}
        \item[]  $w\_id$ \COMMENT{id of node that wrote latest value known to $p$;  initially $\perp$}
        \item[] $\mbox{\it Known\_Writes}[]$ \COMMENT{map from the set of node ids to the powerset of value-timestamp pairs. Initially each entry is $\emptyset$}
	\item[] {\bf Derived Variables:}
	\item[] $\mbox{\it Present}	 = \{q ~|~ enter(q) \in \mbox{\it Server\_Changes} \wedge leave(q) \not\in \mbox{\it Server\_Changes}  \}$
	\item[] $valid\_val$ = value-timestamp pair with latest timestamp that occurs in at least $(f+1)$ elements of $Known\_Writes[]$, else $(\perp,( 0, \perp))$
	\item[] \hrulefill	
	\begin{multicols*}{2}

	\item[]  {\bf When {\sc Enter}$_p$ occurs:}
		\STATE  $\mbox{\it Server\_Changes} := \mbox{\it Server\_Changes} \cup \{enter(p)\}$ 
		\STATE {\bf s-bcast $\langle$``enter'', $p\rangle$}
		\STATE {\bf  c-bcast $\langle$``server-info'', $\mbox{\it Server\_Changes}\rangle$}

	\item[]	
	
	\item[]  {\bf When {\sc Receive}$_p \langle$``enter'', $q\rangle$ occurs:}
	\IF  {\bf IsValidMessage( ``enter",  $q$)}
		\STATE  $\mbox{\it Server\_Changes} := \mbox{\it Server\_Changes} \cup \{enter(q)\}$ 
		\STATE {\bf s-bcast $\langle$``enter-echo'', $\mbox{\it Server\_Changes}$, \\$\mbox{\it Known\_Writes}[p]$, $is\_joined$, $q, p\rangle$	}
		\STATE {\bf  c-bcast $\langle$``server-info'', $\mbox{\it Server\_Changes}\rangle$}
	\ENDIF
	
	\item[]	
	
	\item[]  {\bf When {\sc Receive}$_p \langle$``enter-client'',  $q\rangle$ occurs:}
	\IF  { IsClient($q$)}
					\STATE {\bf c-bcast $\langle$``enter-client-echo'', $\mbox{\it Server\_Changes}$, \\$\mbox{\it Known\_Writes}[p]$,
			$is\_joined$, $q, p\rangle$	}
	\ENDIF

	\item[]	
	
	\item[]  {\bf When {\sc Receive}$_p \langle$``enter-echo'', $C$, $K$, $j$, $q, r\rangle$ \\occurs:}
	\IF  {\bf IsValidMessage( ``enter-echo",  $q$, $r$)}
	
	\STATE $\mbox{\it Server\_Changes} := \mbox{\it Server\_Changes} \cup C$
	\IF{$(j = true)$}
		\STATE  \bf   $\mbox{\it Known\_Writes}[r]:=  \mbox{\it Known\_Writes}[r] \cup K$
	\ENDIF
	
	\IF{$\neg is\_joined \wedge (p = q) $}	
		\STATE  \bf call JoinProtocol($j$)
	\ENDIF
	
	\STATE { \bf call SetValueTimestamp()}
	\ENDIF
	\item[]
	\item[] {\bf When {\sc Receive}$_p \langle$``joined'', $q\rangle$ occurs:}
	
	\IF  {\bf IsValidMessage( ``joined",  $q$)}
		\STATE  $\mbox{\it Server\_Changes} := \mbox{\it Server\_Changes}$\\$ \cup \{enter(q), join(q)\}$ 
		\STATE {\bf s-bcast $\langle$``joined-echo'', $q, p\rangle$}
		\STATE {\bf  c-bcast $\langle$``server-info'', $\mbox{\it Server\_Changes}\rangle$}
	\ENDIF
	\item[]	
	
	\item[] {\bf When {\sc Receive}$_p \langle$``joined-echo'', $q$, $s\rangle$ occurs:}
	\IF  {\bf IsValidMessage( ``joined-echo",  $q$, $s$)}
		\STATE  $\mbox{\it Server\_Changes} := $ \\ $\mbox{\it Server\_Changes} \cup \{enter(q), join(q)\}$ 
		\STATE {\bf c-bcast $\langle$``server-info'', $Server\_Changes\rangle$}
	\ENDIF
	\item[]	
			
	\item[] {\bf When {\sc Leave}$_p$ occurs:}
	\STATE  $\mbox{\it Server\_Changes} := \mbox{\it Server\_Changes} \cup \{leave(p)\}$ 
	\STATE {\bf s-bcast $\langle$``leave'', $p\rangle$}
	\STATE {\bf  c-bcast $\langle$``server-info'', $\mbox{\it Server\_Changes}\rangle$}
	\STATE halt
	\item[]	
	
	\item[] {\bf When {\sc Receive}$_p \langle$``leave'', $q\rangle$ occurs:}
	\IF  {\bf IsValidMessage( ``leave",  $q$)}
		\STATE  $\mbox{\it Server\_Changes} := \mbox{\it Server\_Changes} \cup \{leave(q)\}$ 
		\STATE {\bf s-bcast $\langle$``leave-echo'', $q,p\rangle$}	
		\STATE {\bf c-bcast $\langle$``server-info'', $\mbox{\it Server\_Changes}\rangle$}
	\ENDIF
	\item[]
	
	\item[] {\bf When {\sc Receive}$_p \langle$``leave-echo'', $q$, $s\rangle$ occurs:}
		\IF  {\bf IsValidMessage( ``leave-echo",  $q$, $s$)}
			\STATE  $\mbox{\it Server\_Changes} := \mbox{\it Server\_Changes} \cup \{leave(q)\}$ 
			\STATE {\bf c-bcast $\langle$``server-info'', $\mbox{\it Server\_Changes}\rangle$}	
		\ENDIF

	\item[]
	
	\item[] {\bf When {\sc Receive}$_p \langle$``query'', $rt$, $q\rangle$ occurs:}
	\IF{$is\_joined$  $\wedge$  IsClient($q$)} \label{line:check IsClient1}
		\STATE {\bf c-bcast $\langle$``reply'', $\mbox{\it Known\_Writes}[p], rt , q, p \rangle$}
	\ENDIF
	\item[]
	\item[] {\bf When {\sc Receive}$_p \langle$``update''$, (v,s,i), wt, q \rangle$ occurs:\\}
	\IF{IsClient($q$)}  \label{line:check IsClient2}
	\IF{$(s,i) > (num,w\_id)$}        \label{line:new value?}
		\STATE $(val,num,w\_id) := (v,s,i)$ \label{line:update value}
		\STATE  \bf  $\mbox{\it Known\_Writes}[p]:=  \mbox{\it Known\_Writes}[p] \cup $\\$\{(val, num, w\_id)\}$\label{line:update history}

	\ENDIF
	\IF{$is\_joined$}
	        \STATE {\bf  c-bcast $\langle$``ack'', $wt , q,p \rangle$}
	\ENDIF
	\STATE {\bf s-bcast $\langle$``update-echo'', $\mbox{\it Known\_Writes}[p],p$ $\rangle$}
	\ENDIF
	\item[]
	
	\item[] {\bf When {\sc Receive}$_p \langle$``update-echo'', $K,s \rangle$ occurs:} \\
	\STATE  \bf  $\mbox{\it Known\_Writes}[s]:=  \mbox{\it Known\_Writes}[s] \cup K$
	\STATE \bf call SetValueTimestamp()

\end{multicols*}
\end{algorithmic}
\caption{\AlgName---Code for server $p$.}
\label{algo:Server}
\end{algorithm*}

Algorithm~\ref{algo:Client}: Clients might be in the system from the start or may enter the system at any time. Similar to servers, a newly entered client $p$ runs the JoinProtocol$_p$ procedure in Algorithm~\ref{algo:Common} to join the system. Clients  treat  both read and write operations in a similar manner.
Both operations start with a read phase, which requests the current
value of the register, using a query message, followed by a write
phase, using an update message.  
A write operation broadcasts to all servers the new
value it wishes to write, together with a timestamp, which consists of
a sequence number that is one larger than the largest sequence number
it has seen and its id that is used to break ties. 
A read operation just broadcasts to all servers the value it is about to return, keeping its sequence
number as is. As in~\cite{AttiyaBD1995}, write-back is needed to ensure the
atomicity of read operations.  Both the read phase and the write phase
wait to receive sufficiently many reply messages. 
The fraction $\beta$ \sapta{calculates the number of replies it needs for the operations to terminate. 
as a fraction of the number of nodes it believes are members, i.e., $\beta \cdot |${\it Members}$_p|$.  
This value is stored in the $rw\_bound$ local variable.} 
Setting $\beta$
is also a key challenge in the algorithm as setting it too small might
not return/update correct information from/to the register, whereas
setting it too large might not guarantee termination of the reads and
writes. The fraction $\beta$ also has to ensure that enough replies from correct servers are heard so that these replies can efficiently mask incorrect replies from Byzantine servers.

\begin{algorithm*}[!htbp]
\begin{algorithmic}[1]
\setalglineno{59}
	\small
	\footnotesize
	\item[]{\bf In-built Procedure:}
	\item[] { IsClient($q$) }\COMMENT{returns $true$ if $q$ is a client and $false$ if $q$ is a server } 
	\item[] {\bf Local Variables:}
	\item[] $\mbox{\it Server\_Changes}$ \COMMENT{set that stores information about entering, leaving and joining of servers known by $p$ }	 				\item[] \hspace*{.5in}\COMMENT{initially $\{enter(q) ~|~ q \in S_0\} \cup \{ join(q) ~|~ q \in S_0\}$,  if $p$ is in the system at time 0 and $\emptyset$ otherwise  }
	\item[] $enter\_echo\_counter$ \COMMENT{number of enter-echo messages received so far; initially $0$}
	 \item[] $enter\_echo\_from\_joined\_counter$ \COMMENT{number of enter-echo messages from joined servers received so far; initially 0}
	\item[] $is\_joined$ \COMMENT{Boolean to check if $p$ has joined the system; initially $false$}
	 \item[] $val$ \COMMENT{latest register value known to $p$; initially $\perp$}
	  \item[]  $num$ \COMMENT{sequence number of latest value known to $p$; initially 0}
        \item[]  $w\_id$ \COMMENT{id of node that wrote latest value known to $p$;  initially $\perp$}
        \item[] $\mbox{\it Known\_Writes}[]$ \COMMENT{map from set of node ids to the powerset of value-timestamp pairs. Initially each entry is $\emptyset$}
       \item[] $temp$ \COMMENT{temporary storage for the value being read or written; initially $0$}
	\item[] $tag$ \COMMENT{used to uniquely identify read and write phases of an operation; initially $0$}
	\item[] $rw\_bound$ \COMMENT{the number of replies/acks $p$ should receive before finishing a read/write phase; initially $0$}
	\item[] $rw\_counter$ \COMMENT{the number of replies/acks received so far for a read/write phase; initially $0$}
	 \item[] $rp\_pending$ \COMMENT{Boolean indicating whether a read phase is in progress; initially $false$}	
	\item[] $wp\_pending$ \COMMENT{Boolean indicating whether a write phase is in progress; initially $false$}	
	\item[] $read\_pending$ \COMMENT{Boolean indicating whether a read is in progress; initially $false$}
	\item[] $write\_pending$ \COMMENT{Boolean indicating whether a write is in progress; initially $false$}
	\item[] {\bf Derived Variables:}
	\item[] $\mbox{\it Present}	 = \{q ~|~ enter(q) \in \mbox{\it Server\_Changes} \wedge leave(q) \not\in \mbox{\it Server\_Changes}  \}$
	\item[] $\mbox{\it Members} = \{q ~|~ join(q) \in \mbox{\it Server\_Changes} \wedge leave(q) \not\in \mbox{\it Server\_Changes}\}$
	\item[] $valid\_val$ = value-timestamp pair with latest timestamp that occurs in at least $(f+1)$ elements of $Known\_Writes[]$, else $(\perp,( 0, \perp))$ 
	
	\item[] \hrulefill	
	\begin{multicols*}{2}
	\item[]  \bf When {\sc Enter}$_p$ occurs:
	\STATE s-bcast $\langle$``enter-client'', $p\rangle$
	\item[]	
	\item[] { \bf When {\sc Receive}$_p \langle$``enter-client-echo'', $C$, $K$, $j$, $q, r \rangle$ occurs:}
	\IF  {\bf IsValidMessage( ``enter-client-echo",  $q$)\\ $\wedge$ $(p=q)$}
		\STATE $\mbox{\it Server\_Changes} := \mbox{\it Server\_Changes} \cup C$
		\IF{$(j = true)$}
			\STATE { \bf  $\mbox{\it Known\_Writes}[r]:=  \mbox{\it Known\_Writes}[r] \cup K$}
		\ENDIF
		\IF{$\neg is\_joined \wedge (p = q) $}
			\STATE  {\bf call JoinProtocol($j$)}
		\ENDIF
	\ENDIF

	\STATE {\bf call SetValueTimestamp()}
	
	\item[]
	\item[] \bf When {\sc Receive}$_p \langle$``server-info'', $C \rangle$ occurs:
			\STATE $\mbox{\it Server\_Changes} := \mbox{\it Server\_Changes} \cup C$
			\item[]	
	\item[] {\bf Procedure} BeginReadPhase()	
	\STATE $tag$++
	\STATE {\bf s-bcast $\langle$``query'', $tag$, $p\rangle$ } \label{line:bcast query}
	\STATE $rw\_bound :=  \beta \cdot |\mbox{\it Members}|$
	\STATE $rw\_counter := 0$
	\STATE $rp\_pending := true$	
	\item[]	
	
	\item[] {\bf When {\sc Receive}$_p \langle$``reply'', $K, rt, q, s \rangle$ occurs:}
	\IF  {\bf IsValidMessage( ``reply",  $q$, $rt$, $s$)}
		\IF{$rp\_pending  \wedge (rt = tag) \wedge (q=p)$}
			\STATE $rw\_counter$++         \label{line:inc heard from resp}
			\STATE { \bf  $\mbox{\it Known\_Writes}[s]:=  \mbox{\it Known\_Writes}[s] \cup K$}
			\IF{$rw\_counter \ge rw\_bound$}   \label{line:quorum reached}
				\STATE {\bf call SetValueTimestamp()	}		
				\STATE $rp\_pending := false$
				\STATE {\bf call BeginWritePhase()}
			\ENDIF
		\ENDIF
	\ENDIF

	\item[]
	\item[] {\bf Procedure} BeginWritePhase()
	\IF{$write\_pending$}		
		\STATE  $val := temp$  \label {line:new timestamp1}
		\STATE $num$++\label {line:new timestamp}
		\STATE $w\_id := p$ \label {line:new timestamp2}
	
	\ENDIF
	\IF{$read\_pending$}
	\STATE $temp:= val$
	\ENDIF
	\STATE s-bcast $\langle$``update'', $(temp,num,w\_id)$,$tag, p \rangle$ \label{line:bcast update1} \
	\STATE $rw\_bound :=  \beta \cdot |\mbox{\it Members}| $
	\STATE $rw\_counter := 0$
	\STATE $wp\_pending := true$
	\item[]
	
	\item[] \bf When {\sc Receive}$_p \langle$``ack'', $wt, q, s \rangle$ occurs:
	\IF  {\bf IsValidMessage( ``ack",  $q$, $wt$, $s$)}
		\IF{$wp\_pending \wedge (wt = tag) \wedge (q=p)$}
			\STATE $rw\_counter$++                \label{line:inc heard from ack}
			\IF{$rw\_counter \ge rw\_bound$}
				\STATE $wp\_pending := false$
				\IF{$read\_pending$}
					\STATE $read\_pending := false$
					\STATE generate {\sc Return}($temp$) response			
				\ENDIF			
				\IF{$write\_pending$}
					\STATE $write\_pending := false$
					\STATE generate {\sc Ack} response
				\ENDIF
			\ENDIF
		\ENDIF
	\ENDIF

	\item[]
	\item[] {\bf When {\sc Leave}$_p$ occurs:}
	\STATE halt
			
\end{multicols*}
\end{algorithmic}
\caption{\AlgName---Code for client, $p$.}
\label{algo:Client}
\end{algorithm*}

Algorithm~\ref{algo:Common}: The JoinProtocol() procedure helps newly entered nodes to join the system. The other procedures in this algorithm are used to deal with Byzantine servers and their arbitrary nature. The procedure SetValueTimestamp() checks and updates the value-timestamp  triple ($(val, (num,w\_id))_p$) to $valid\_val_p$ if the timestamp of  $valid\_val_p$ is higher than the latest known $(num,w\_id)_p$ pair. The variable $valid\_val_p$ is necessary to make sure that before writing any value (learned from other servers) to the register, the value was seen by at least $f+1$ servers.
A Byzantine server $p$ may send out more than one reply for a given message or keep replying after it has sent out a $leave_p$ message.  The three IsValidMessage() procedures deal with these situations. They check to make sure that only one reply from each server for a message is processed by all nodes. They also check  whether the sender $q$ has already sent a $leave_q$ message.  Reads and writes invoked by Byzantine servers are ignored by correct servers by the  IsClient() checks on Lines~\ref{line:check IsClient1} and \ref{line:check IsClient2} in Algorithm~\ref{algo:Server}.

\begin{algorithm*}[hptb]
\begin{algorithmic}[1]
\setalglineno{116}
	\small
	\footnotesize
	\item[] 
	\begin{multicols*}{2}

	\item[] \bf Procedure JoinProtocol($j$)
	\STATE $enter\_echo\_counter++$
	\IF{$(j = true) \wedge (join\_bound = 0)$}
			\STATE $enter\_echo\_from\_joined\_counter++$

		\IF{$ enter\_echo\_from\_joined\_counter > f $}\label{line: f+1 non faulty}
			\STATE 	$join\_bound :=  \gamma \cdot | \mbox{\it Present} | $\label{line:calculate join bound}
		\ENDIF
	\ENDIF
	\IF{$enter\_echo\_counter\!\geq\! join\_bound\! >\! 0$} \label{line:check if enough enter echoes}
			\STATE $is\_joined := true$
			
			\IF{ { IsClient($p$)}}
				\STATE {\bf generate {\sc Joined}$_p$ response}
				\ELSE
				\STATE  $\mbox{\it Server\_Changes} := \mbox{\it Server\_Changes} \cup \{join(p)\}$ 
				\STATE {\bf s-bcast $\langle$``joined'', $p\rangle$}
				\STATE {\bf  c-bcast $\langle$``server-info'', $\mbox{\it Server\_Changes}\rangle$}\label{line: server joined}
			\ENDIF

	\ENDIF
	\item[]
	\item[]  {\bf Procedure SetValueTimestamp()}
			
	\IF{$valid\_val \neq \perp$} \label{line:valid val}
	\IF{\bf timestamp of $valid\_val > (num,w\_id)$}
		\STATE \bf $(val,num,w\_id) := valid\_val$
		\STATE  \bf  $\mbox{\it Known\_Writes}[p]:=  \mbox{\it Known\_Writes}[p] \cup $\\$\{(val, num, w\_id)\}$
	\ENDIF		
	\ENDIF
	\item[]
	 \item[]  {\bf Procedure IsValidMessage( $type$, $r$)}
	\IF{$type$ = (``enter"$ \lor$ ``joined"$\lor$``leave") $ \wedge $ $(leave(r) \notin \mbox{\it Server\_Changes})$}
		\STATE  {\bf return $true$ if  this is the first  $type$ message received from $r$, else return $false$}
	\ENDIF	
	\item[]
	\item[]  {\bf Procedure IsValidMessage( $type$, $q$, $r$)}
	\IF{$type$ = (``enter-echo" $ \lor$ ``enter\_client-echo" $ \lor$ \\ \hspace*{.5in} ``joined-echo" $ \lor$  ``leave-echo")  \\ \hspace*{.5in}$\wedge $ $(leave(r) \notin \mbox{\it Server\_Changes})$}
		\STATE {\bf return $true$ if this is the first $type$ message for $q$ received from $r$, else return $false$}
	\ENDIF		
	\item[]
	\item[]  {\bf Procedure IsValidMessage( $type$, $q$, $tag$, $r$)}			
	\IF{$type$ = (``reply" $ \lor$ ``ack") $\wedge $ $(leave(r) \notin \mbox{\it Server\_Changes})$ }
		\STATE {\bf  return $true$ if  this is the first  $type$ message for $q$ with sequence $tag$ received from $r$, else return $false$}
	\ENDIF

\end{multicols*}
\end{algorithmic}
\caption{\AlgName---Procedures used by client/server $p$}
\label{algo:Common}
\end{algorithm*}

	The correctness of \AlgName{} relies on the system parameters
$\alpha$, $f$, and $NS_{min}$ satisfying the following constraints,
for some choice of algorithm parameters $\beta$ and $\gamma$:

\begin{align}
\alpha  &\leq 1 - 2^{-1/4} \approx 0.159   \label{constraint:A }
\\ 
1       &\le (1 - \alpha)^3 NS_{min} -2f   \label{constraint:B }
\\
\gamma  &\geq \frac{1+2f}{(1-\alpha)^3NS_{min}}  +  \frac{(1 + \alpha)^3}{(1 - \alpha)^3} - 1   \label{constraint:C }
   \displaybreak[0]
\\
\gamma &\leq \frac{(1 - \alpha)^3}{(1 + \alpha)^3} - \frac{f}{(1+\alpha)^3 NS_{min}} \label{constraint:D }
   \displaybreak[0]
\\
\beta  & \leq \frac{(1 - \alpha)^3}{(1 + \alpha)^2}  - \frac{f}{(1+\alpha)^2 NS_{min}} \label{constraint:E }
   \displaybreak[0]
\\
\beta  &> \frac{ (1 + \alpha)^5 - 1 +2f/NS_{min} }{(1 - \alpha)^4 -f/NS_{min}}  \label{constraint:F }
   \displaybreak[0]
\\
\beta  &> \frac{(1 + \alpha)^3 - (1 - \alpha)^3 +1+ (1+3f)/NS_{min}}{\left[(2 + 2\alpha + \alpha^2)(1 - \alpha)^2(1 + \alpha)^{-2}\right] -2f/NS_{min}}  \label{constraint:G }
\end{align}

\begin{table}
\begin{center}
\begin{tabular}{ |c|c|c|c|c| }
\hline
\multicolumn{3}{|c|}{system } &\multicolumn{2}{c|}{ algorithm } \\
\multicolumn{3}{|c|}{ parameters} &\multicolumn{2}{c|}{ parameters} \\\hline
maximum & minimum & churn   &   $join\_bound$  & $rw\_bound$   \\
 failures &  system & rate  &     fraction  &  fraction  \\
($f$) &  size ($NS_{min}$) & ($ \alpha$) &    ($\gamma$) &($\beta$)   \\\hline
$1$ &$8$ &  $0$& N/A & $0.86$ \\ \hline
$1$ &  $10$ & $0.01$ & $0.82$ & $0.84$ \\ \hline
$1$ &$13$ &  $0.02$ & $0.79$ & $0.80$ \\ \hline
$1$ &$190$ &  $0.05$ & $0.79$ & $0.80$ \\ \hline
$2$ &$19$ &  $0.01$ & $0.80$ & $0.83$ \\ \hline
$2$ &$24$ &  $0.02$ & $0.81$ & $0.82$ \\ \hline
$2$ &$347$ &  $0.05$& $0.70$ & $0.77$ \\ \hline
$5$ &$44$ &  $0.01$ & $0.80$ & $0.83$ \\ \hline
$5$ &$57$ &  $0.02$ & $0.79$ & $0.82$ \\ \hline
$5$ &$826$ &  $0.05$& $0.79$ & $0.82$ \\ \hline
$10$ &$85$ &  $0.01$ & $0.80$ & $0.83$ \\ \hline
$10$ &$113$ &  $0.02$ & $0.79$ & $0.82$ \\ \hline
$10$ &$1630$ &  $0.05$& $0.79$ & $0.82$ \\ \hline
$100$ &$838$ & $0.01$ & $0.79$ & $0.82$ \\ \hline
$100$ &$1107$ &  $0.02$ & $0.79$ & $0.82$ \\ \hline
$100$ &$16015$ &  $0.05$& $0.79$ & $0.82$ \\ \hline
$1000$ &$8360$ &  $0.01$ & $0.79$ & $0.82$ \\ \hline
$1000$ &$11042$ &  $0.02$ & $0.79$ & $0.82$ \\ \hline
$1000$ &$159935$ &  $0.05$& $0.79$ & $0.82$ \\ \hline

\end{tabular}
\end{center}
\caption{Values for the parameters 
that satisfy constraints (1)--(7).}
\label{table:assumptions}
\end{table}

Constraint~(\ref{constraint:A }) is an upper bound on the churn rate 
to ensure that not too many servers  can
leave the system in an interval of length $4D$.
Constraint~(\ref{constraint:B }) is a lower bound on the minimum system size 
to ensure that at least $f+1$ correct servers are in the system throughout an interval of length $3D$ encompassing the time a node enters, thus ensuring that the newly entered node successfully terminates its joining protocol.
Constraint~(\ref{constraint:C }) ensures that the $join\_bound$
fraction, $\gamma$, is large enough such that updated information
about the system is obtained by an entered node before it joins the
system. Constraint~(\ref{constraint:D }) ensures that $\gamma$ is small
enough such that for all entered nodes, a join operation terminates if
the entered node \sapta{is not faulty and does not leave.} 
Constraint~(\ref{constraint:E }) ensures that the $rw\_bound$ fraction, $\beta$, is small enough such that termination of reads and writes is guaranteed.  Constraints~(\ref{constraint:F }) and~(\ref{constraint:G }) ensure that $\beta$ is large enough such that atomicity is not violated by read and write operations.
Table~\ref{table:assumptions} gives a few sets of values for which the
above constraints are satisfied. In all consistent sets of parameter values, the churn rate $\alpha$ is never more than $0.05$ \sapta{ and $NS_{min} > 8.5f$. The algorithm can tolerate any size of $f$ as long as $NS_{min}$ is proportionally big.} 

 {\sc \AlgName{}} violates atomicity if Assumption A5
is violated. 

\section{Correctness proof of  {\sc \AlgName{}}} \label{section:proof}

We will show that {\sc \AlgName{}} satisfies the  properties C1 to C3 listed
at the end of Section~\ref{section:model}.
Lemmas~\ref{lem:size} through~\ref{lem:knowswhenjoined} 
are used to prove
Theorem~\ref{thm:joins}, 
which states that every client and any correct server eventually joins, provided it does not
crash or leave.  Lemmas~\ref{lem:present-2D} through~\ref{lem:joined-for-D}
are used to prove
Theorem~\ref{thm:ops-live}, which states
that every operation invoked by a client that remains active eventually
completes.  Lemmas~\ref{lem:lin1} through~\ref{lem:lin4} are used to prove Theorem~\ref{thm:atomicity}, which
states that atomicity is satisfied.

Consider any execution.
We begin by bounding the number of servers that enter
during an interval of time
and the number of servers that are present at the end of the
interval, as compared to the number present at the beginning.

\sapta{ \begin{lemma}\label{lem:size}
For all $i \in \nats$ and all $t \geq 0$, at most $((1 + \alpha)^i-1)NS(t)$ servers enter during $(t,t+Di]$ and $(1 - \alpha)^i NS(t) \leq NS(t+Di) \leq (1+\alpha)^iNS(t)$.
\end{lemma}
\begin{proof}
The proof is by induction on $i$ and is adapted from~\cite{AttiyaCEKW2018}.
For $i = 0$ and all $t \geq 0$, $(t,t+Di]$ is empty,
and hence, $0 = ((1+\alpha)^i-1)NS(t)$ servers enter during this interval and
\[ NS(t+iD) = NS(t) =  (1+\alpha)^iNS(t) = (1 - \alpha)^i NS(t) . \]
Now let $i \geq 0$ and $t \geq 0$.
Suppose
at most $((1 + \alpha)^i-1)NS(t)$ servers enter during $(t,t+Di]$ and
$(1 - \alpha)^i NS(t) \leq NS(t+Di) \leq (1+\alpha)^iNS(t)$.

Let $e\geq 0$ and $\ell\geq 0$ be the number of servers that enter and leave,
respectively, during $(t+Di,t+D(i+1)]$.
By Assumption A5,
 $e+\ell\leq \alpha NS(t+Di)$, so
$e,\ell \leq \alpha NS(t+Di) \leq \alpha(1+\alpha)^iNS(t)$.
The number of servers that enter during $(t,t+D(i+1)]$ is at most
\begin{multline*}
((1 + \alpha)^i-1)NS(t) + e
    \leq ((1 + \alpha)^i-1)NS(t) + \\ \alpha(1+\alpha)^iNS(t) \\
      = ((1+\alpha)^{i+1}-1)NS(t) .
\end{multline*}
Hence,
\begin{multline*}
NS(t + D(i+1)) \leq NS(t) +  ((1+\alpha)^{i+1}-1)NS(t)\\
 = (1+\alpha)^{i+1}NS(t) .
\end{multline*}
Furthermore,
\begin{multline*}
NS(t + D(i+1)) \geq NS(t+Di) -\ell\\
    \geq NS(t+Di) - \alpha NS(t+Di) \\
    = (1 - \alpha)NS(t+Di) \geq (1-\alpha)^{i+1}NS(t) .
\end{multline*}
By induction, the claim is true for all $i \in \nats$.
\end{proof}

We are also interested in the number of servers that leave
during an interval of time.
The calculation of the maximum number of servers that leave during an interval
is complicated by the possibility of servers entering during the interval,
allowing additional servers to leave.

\begin{lemma}\label{lem:left2}
For $\alpha >0$, all nonnegative integers $i \leq -1/\log_2(1-\alpha)$
and every time $t \geq 0$,
at most $(1-(1-\alpha)^i) NS(t)$ servers leave during $(t,t+Di]$.
\end{lemma}
}

\begin{proof}
The proof is by induction on $i$ and is adapted from~\cite{AttiyaCEKW2018}.
When $i = 0$, the interval is empty,
so $0 = (1-(1-\alpha)^0) NS(t)$ servers leave during the interval.
Now let $i \geq 0$, let $t \geq 0$, and suppose
at most $(1-(1-\alpha)^i) NS(t+D)$ servers leave during $(t+D,t+D(i+1)]$.

Let $e\geq 0$ and $\ell\geq 0$ be the number of servers that enter and leave, respectively, during $(t,t+D]$.
By Assumption A5,
 $e+\ell \leq \alpha NS(t)$, so $\ell \leq  \alpha NS(t)$ and
$NS(t+D) = NS(t)+ e - \ell  =  NS(t) + (\ell + e) - 2\ell \leq (1 + \alpha) NS(t) - 2\ell$.
The number of servers that leave during $(t,t+D(i+1)]$
is the number that leave during $(t,t+D]$ plus the number that leave during $(t+D,t+D(i+1)]$,
which is at most

\begin{multline*}
\ell + (1-(1-\alpha)^i) NS(t+D) \\
\leq  \ell + (1-(1-\alpha)^i) [(1 + \alpha) NS(t) - 2\ell] \\
=  (1-(1-\alpha)^i) (1 + \alpha) NS(t) + ( 2(1-\alpha)^i-1) \ell\\
\leq  (1-(1-\alpha)^i) (1 + \alpha) NS(t) + ( 2(1-\alpha)^i-1)\alpha NS(t)\\
= (1-(1-\alpha)^{i+1}) NS(t).
\end{multline*}

Note that $2(1-\alpha)^i -1 \geq 0$, since $i \leq -1/\log_2(1-\alpha)$.
By induction, the claim is true for all $i \in \nats$.
\end{proof}


 Lemma~\ref{lem:gen0} proves that at least $f+1$ correct servers are active throughout any interval of length $3D$. This lemma is necessary to ensure that at all times, an active node (client or server) that expects replies, hears back from at least $f+1$ correct servers in order to mask the bad information sent by Byzantine servers. 
\sapta{
\begin{lemma}
\label{lem:gen0}
For every $t > 0$, at least $f+1$ correct servers
are active throughout $[\max\{0,t -2D\},t +D]$.
\end{lemma}
}

\begin{proof}
Let $S$ be the set of servers present at time
$t' = \max\{0,t -2D\}$,
so $|S| =  NS(t') \geq NS_{min}$.
Constraint~(\ref{constraint:A }) implies that
$-1/\log_2(1- \alpha) \geq 4 \geq  3$. So,
by Lemma~\ref{lem:left2},
at most $(1-(1-\alpha)^3) |S|$ servers leave during
$(t',t+D]$
and there are at least $(1-\alpha)^3 |S|$ servers present throughtout  time interval $(t',t+D]$. At any point in time, there are at most $f$ Byzantine servers in the system.  
Thus, at least

\begin{equation*}
(1-\alpha)^3 |S| -f  \geq (1-\alpha)^3    NS_{min} -f
\end{equation*}

 correct servers in $S$ are active at time $t+D$.
By Constraint~(\ref{constraint:B }),
$ (1- \alpha)^3 NS_{min} -f   \ge f+1$, 
so at least   $f+1$ correct servers in $S$ are still active at time $t+D$.
\end{proof}

Below, a local variable name is superscripted with $t$ to denote
the value of that variable at time $t$; e.g., $v_p^t$ is the value
of node $p$'s local variable $v$ at time $t$.

\sapta{
In the analysis, we frequently compare the data in
nodes' $\mbox{\it Server\_Changes}$ sets to the set of {\sc Enter}, {\sc Joined}, and {\sc Leave}
events that have actually occurred.
To facilitate this comparison, we define a set {\it SysInfo}$^I$ that contains
perfect information about correct servers for the time interval $I$.
For each server $q$, let $t_q^e$,  and $t_q^\ell$ be the times
when the events {\sc Enter}$_q$  and {\sc Leave}$_q$
occur, and let $t_q^j$ be the time when server $q$ sends out a joined message.
 Similarly, for each client  $q$, let $t_q^e$, $t_q^j$, and $t_q^\ell$ be the times
when the events {\sc Enter}$_q$, {\sc Joined}$_q$, and {\sc Leave}$_q$
occur, respectively.

Recall that $S_0$ is the set of servers that were in the system initially. If $q \in S_0$, then we set $t_q^e = t_q^j = 0$.
Then we have:
\begin{multline*}
\mbox{\it SysInfo}^I =
    \{ enter(q) ~|~t_q^e \in I\} 
    \cup \{ join(q) ~|~t_q^j \in I\} \cup \\ \{ leave(q) ~|~t_q^\ell \in I\} . 
\end{multline*}
In particular,
\[
\mbox{\it SysInfo}^{[0,0]} = \{ enter(q) ~|~q \in S_0\}
                        \cup  \{ join(q) ~|~q \in S_0\} .
\]
}

Since a client or correct server $p$ that is active throughout $[t_p^e,t+D]$ directly
receives all enter, joined, and leave messages broadcast by active clients or correct servers during $[t_p^e,t]$,
within $D$ time, we have:

\begin{observation}
\label{obs:S0}
For every client and any  correct server $p$ and all times $t \geq t_p^e$, if $p$ is active at time $t + D$, then
\knows{p}{t+D}{[t_p^e, t]}.
\end{observation}
Let $C_0$ be the set of clients that are in the system initially.
By assumption, for every node $p \in S_0  \cup C_0$,
$\mbox{\it SysInfo}^{[0,0]}  \subseteq \mbox{\it Server\_Changes}_p^0$,
and hence Observation~\ref{obs:S0} implies:

\begin{observation}
\label{obs:S1}
For every client and any correct server $p \in S_0 \cup C_0$, if $p$ is active at time $t \geq 0$, then
\knows{p}{t}{[0, \max\{0,t-D\}]}.
\end{observation}

The purpose of Lemmas~\ref{lem:gen1} to~\ref{lem:knowswhenjoined} show that information about correct 
servers entering, joining, and leaving is propagated to active clients and correct servers properly, via the
{\it Server\_Changes} sets.

\sapta{
\begin{lemma}
\label{lem:gen1}
Suppose that, at time $T''$, a client or correct server $p \notin S_0 \cup C_0$ receives
an enter-echo message from a correct server $q$ sent at time $T'$
in reply to an enter message from $p$.
Let $T$ be any time such that $\max\{0,T''-2D\} \le T \le t_p^e$.
Suppose $p$ is active at time $T + 2D$ and
$q$ is active throughout $[U,T+D]$, where $U \le \max\{0,T''-2D\}$.
Then {\it SysInfo}$^{(U,T]} \subseteq$ {\it Server\_Changes}$_p^{T+2D}$.
\end{lemma}
}

\begin{proof}
The proof is adapted from~\cite{AttiyaCEKW2018} to include Byzantine servers. 
Consider any node $r$ that enters, joins, or leaves at time $\hat{t} \in
(U,T]$.
Note that $q$ directly receives this event's announcement,
since $q$ is active throughout $(U,T+D]$, which contains $[\hat{t},\hat{t}+D]$,
the interval during which the announcement message is in transit.
There are two cases, depending on the time, $v$, at which $q$ receives this message.
\begin{itemize}
\item[Case 1:]
$v \leq T'$.
Since $q$ receives the enter message from $p$ at $T'$,
information about this change to $r$ is in
{\em Server\_Changes}$_q^{T'}$, in the enter-echo message
that $q$ sends to $p$ at time $T'$.
Thus, this information is in
{\em Server\_Changes}$_p^{T''} \subseteq$ {\em Server\_Changes}$_p^{T+2D}$.
\item[Case 2:]
$v > T'$.
Messages are not received before they
are sent, so $T' \ge t_p^e$.
Since $v \le \hat{t} + D$, it follows that $v + D \le \hat{t} + 2D \leq T+2D$.
Thus $[v,v+D]$ is contained in $[t_p^e,T+2D]$.
Immediately after receiving the announcement about $r$,
server $q$ broadcasts an echo message in reply. 
Since $p$ is active throughout  this interval,
it directly receives this echo message.
\end{itemize}
In both cases,
the information about $r$'s change reaches $p$ by time $T + 2D$.
It follows that {\em SysInfo}$^{(U,T]} \subseteq$ {\em Server\_Changes}$_p^{T+2D}$.
\end{proof}

\sapta{
\begin{lemma}
\label{lem:gen2}
For every client and any correct server $p$, if $p$ is active at time $t \geq t_p^e + 2D$,  then \knows{p}{t}{[0,t-D]}.
\end{lemma}
}

\begin{proof}
The proof is adapted from~\cite{AttiyaCEKW2018} to include Byzantine servers. 
The proof is by induction on the order in which nodes enter the system.
If $p \in S_0\cup C_0$, then $t_p^e = 0$, so \knows{p}{t}{[0,t-D]} follows from Observation~\ref{obs:S1}.

Now consider any node $p \not\in S_0\cup C_0$ and suppose that the claim holds for all nodes that
enter earlier than $p$. Suppose $p$ is active at  time $t \geq t_p^e+ 2D$.
By Lemma~\ref{lem:gen0}, there is at least \sapta{$f+1$ servers (let $q$ be one of these)} that are active throughout $[\max\{0,t_p^e -2D\},t_p^e +D]$.
Server $q$ receives an enter message from $p$ at some time $t' \in [t_p^e, t_p^e +D]$
and sends an enter-echo message back to $p$.
This message is received by $p$ at some time $t'' \in [t',t'+D]$.

If $q \in S_0$, then
\knows{q}{t'}{[0,\max\{0,t'-D\}]}, by Observation~\ref{obs:S1}.
If $q \not\in S_0$, then $0 < t_q^e \leq \max\{0,t_p^e -2D\}$, so
$t_q^e \leq  t_p^e - 2D$. Therefore
$t_q^e + 2D \leq t_p^e \leq t'$.
Since $q$ entered earlier than $p$, it follows from the induction hypothesis
that  \knows{q}{t'}{[0,t'-D]}.
Thus, in both cases, \knows{q}{t'}{[0,\max\{0,t'-D\}]}.
At time $t''\leq t$, $p$ receives the enter-echo message from $q$, so
\knows{p}{t''}{[0,\max\{0,t'-D\}]} $\subseteq \mbox{\it Server\_Changes}_p^{t}$.

Applying Lemma~\ref{lem:gen1} for $q$, with
$U = \max\{0,t_p^e -D\}$, $T = t_p^e$, $T' =t'$ and $T'' = t''$ implies
\[
\mbox{\it SysInfo}^{(\max\{0,t'-D\},t_p^e]} \subseteq \mbox{\it Server\_Changes}_p^{t_p^e + 2D} .
\]
Since $t \geq t_p^e+2D$, $\mbox{\it Server\_Changes}_p^{t_p^e+2D}$ is a subset of $\mbox{\it Server\_Changes}_p^{t}$.
Observation~\ref{obs:S0} implies \knows{p}{t}{[t_p^e,t-D]}.
Hence,  \knows{p}{t}{[0,t-D]}.
\end{proof}

\sapta{
\begin{lemma}
\label{lem:knowswhenjoined}
For every client and any correct server $p \not\in S_0\cup C_0$, if $p$ joins at time $t_p^j$ and
is active at time $t \geq  t_p^j$, then \knows{p}{t}{[0,\max\{0,t-2D\}]}.
\end{lemma}
}

\begin{proof}
The proof is by induction on the order in which clients and correct servers join the system.
Let $p \not\in S_0\cup C_0$ be a client or correct server that joins at time $t_p^j \leq t$
and suppose the claim holds for all clients and correct servers that join before $p$.
If $t \geq t_p^e +2D$, then the claim follows by Lemma~\ref{lem:gen2}.
So, suppose $t < t_p^e +2D$.

Before joining, $p$ receives $f+1$ enter-echo message from  joined servers in reply to its enter message (Line number~\ref{line: f+1 non faulty}). Out of these, at most $f$ can be from Byzantine servers. Thus, at least one reply is from a correct server. 
Suppose $p$  receives the first  enter-echo message at time $t''$  
sent by correct server $q$ at time $t'$; $t_p^e \leq t' \leq t'' \leq t_p^j$.
From Lemma~\ref{lem:gen2}, we know that this message from  $q$ has a perfect information about the $\mbox{\it Server\_Changes}^{t'-2D}$ set. This in turn means that it has perfect information about the derived set {\em Present}$^{t'-2D}$. Byzantine servers can only modify the information about the $\mbox{\it Server\_Changes}$ set by sending a subset of its $\mbox{\it Server\_Changes}$ set. So, when node $p$ receives at least one reply is from a correct server, the incomplete information sent by Byzantine servers is overshadowed by this one reply from $q$ and thus $p$ has a perfect information about  $\mbox{\it Present}^{t'-2D}$. 

If correct server $q \in S_0$, then by Observation~\ref{obs:S1},
\knows{q}{t'}{[0,\max\{0,t'-D\}]}.
Otherwise, by the induction hypothesis, \knows{q}{t'}{[0,\max\{0,t'-2D\}]},
since $q$ joined prior to $p$ and is active at time $t' \geq t_q^j$.
Note that  $\mbox{\it Server\_Changes}_q^{t'}$
 $\subseteq \mbox{\it Server\_Changes}_p^{t''}$
 $\subseteq \mbox{\it Server\_Changes}_p^{t}$.
If $t \leq 2D$, then $\max\{0,t-2D\} = 0$ and the claim holds.

If $t > 2D$, then let $S$ be the set of servers present at time $\max\{0,t'-2D\}$;
$|S|= NS(\max\{0,t'-2D\})$.
By Lemma~\ref{lem:left2} and Constraint~(\ref{constraint:A }),
at most $(1-(1-\alpha)^3)|S|$ servers leave during $(\max\{0,t'-2D\},t'+D]$.
Since $t'' \leq t'+D$, it follows that
$|\mbox{\it Present}_p^{t''}| \geq |S| - (1-(1-\alpha)^3)|S| = (1-\alpha)^3|S|$.
Hence, from lines~\ref{line:calculate join bound}
and~\ref{line:check if enough enter echoes}
of Algorithm~\ref{algo:Common}, $p$ waits until it has received at least $join\_bound = \gamma \cdot |\mbox{\it Present}_p^{t''}|  \geq \gamma \cdot (1-\alpha)^3|S| $ enter-echo messages before joining.

By Lemma~\ref{lem:size},
at most $((1+\alpha)^3-1)|S|$ servers enter during $(\max\{0,t'-2D\},t'+D]$.
Thus, at time $t'+D$, at most $(1+\alpha)^3|S|$ servers are present, at most $f$ of which are Byzantine.

Hence, the number of enter-echo messages $p$ receives before joining from servers that were active throughout $[\max\{0,t'-2D\},t'+D]$ is
$join\_bound$ minus the total number of server enters, leaves and faults (as Byzantine servers may not reply at all),
which is at least

\begin{multline}
\gamma \cdot (1-\alpha)^3 |S| - 
    [((1+\alpha)^3 -1)|S| + (1-(1-\alpha)^3)|S|  +f] \\
= [(1+\gamma) (1-\alpha)^3 - (1+\alpha)^3]|S| -f \\
\geq  [(1+\gamma) (1-\alpha)^3 - (1+\alpha)^3]NS_{min} -f \label{Lemma9equation}
\end{multline}

Rearranging Constraint~(\ref{constraint:C }), we get
\[
[(1+\gamma)(1- \alpha)^3 -(1+\alpha)^3]NS_{min}-f \geq f+1 ,
\]
so expression~(\ref{Lemma9equation}) is at least $f+1$.
Hence $p$ receives an enter-echo message at some time $T'' \leq t_p^j$
from a correct server $q'$ that is active throughout
\[
[\max\{0,t'-2D\},t'+D] \supseteq [\max\{0,t'-2D\},t-D] .
\]
Let $T'$ be the time that $q'$ sent its enter-echo message in reply 
 to
the enter message from $p$. 
Applying Lemma~\ref{lem:gen1} for $q'$, with
$U = \max\{0,t' -2D\}$, and $T = t-2D$
gives {\em SysInfo}$^{(\max\{0,t'-2D\},t-2D]} \subseteq$ {\it Server\_Changes}$_p^t$.

Thus, we get {\em SysInfo}$^{[0,t-2D]} =$ {\em SysInfo}$^{[0,\max\{0,t'-2D\}]} \cup$ {\em SysInfo}$^{(\max\{0,t'-2D\},t-2D]}
\subseteq$ {\it Server\_Changes}$_p^t$.
\end{proof}

Lemmas~\ref{lem:size} through~\ref{lem:knowswhenjoined}   
are used to prove Theorem~\ref{thm:joins} as follows: 

\begin{theorem} \label{thm:joins}
Every client and any correct server $p \not\in S_0\cup C_0$ that
is active for at least $2D$ time
     after it enters succeeds in joining. 
\end{theorem}

\begin{proof}
The proof is by induction on the order in which clients and correct servers enter the system.
Let $p \not\in S_0\cup C_0$ be a client or correct server that enters at time $t_p^e$ and is active at time $t_p^e + 2D$.
Suppose the claim holds for all client and correct servers that enter before $p$.

By Lemma~\ref{lem:gen0}, there  are $f+1$ correct servers that are active throughout
$[\max\{t_p^e -2D,0\},t_p^e +D]$. Let $q$ be one such server.
If $q \in S_0$, then $q$ joins at time 0.
If not, then $t_q^e \leq t_p^e-2D$,
so, by the induction hypothesis, $q$ joins by time
$t_q^e + 2D \leq  t_p^e$.
Since $q$ is active at time  $t_p^e+D$,
it receives the enter message from $p$ during $[t_p^e, t_p^e+D]$
and sends an enter-echo message in reply. 
Since $p$ is active at time $t_p^e+2D$,
it receives the enter-echo message from $q$ by time $t_p^e+2D$.
Hence, by time $t_p^e+2D$, $p$ receives at least one enter-echo message
from a correct joined server in reply to its enter message.

Suppose the first enter-echo message $p$ receives from a correct joined server in reply 
to its enter message is sent by server $q'$ at time $t'$ and received by $p$ at time $t''$.
By Lemma~\ref{lem:knowswhenjoined}, \knows{q'}{t'}{[0,\max\{0,t'-2D\}]} $\subseteq \mbox{\it Server\_Changes}_p^{t''}$.

Let $S$ be the set of servers present at time $\max\{0,t'-2D\}$.
Since $t'' \leq t'+D$, it follows from Lemma~\ref{lem:size} that at most $( (1+\alpha)^3 -1)|S|$
servers enter during $(\max\{0,t'-2D\},t'']$. Thus, $|\mbox{\it Present}_p^{t''}| \leq|S| +( (1+\alpha)^3 -1)|S|  = (1+\alpha)^3 |S|$.
From line~~\ref{line:calculate join bound} in Algorithm~\ref{algo:Common}, it follows that
$join\_bound \leq \gamma \cdot (1 +\alpha)^3|S|$.

By Lemma~\ref{lem:left2} and Constraint~(\ref{constraint:A }),
at most $(1-(1-\alpha)^3)|S|$ servers leave during $(\max\{0,t'-2D\},t'+D]$.
At most $f$ servers are Byzantine  
at $t'+D$.
Since $t_p^e \leq t' \leq t_p^e +D$,
the servers in $S$ that do not leave during $(\max\{0,t'-2D\},t'+D]$
and are not Byzantine at $t'+D$ are active throughout $[t_p^e,t_p^e+D]$
and send enter-echo messages in reply
 to $p$'s enter message.
By time $t_p^e+2D$, $p$ receives all these enter-echo messages.
There are at least

\[|S| - (1- (1-\alpha)^3) |S|-f = (1-\alpha)^3 |S|-f\]
such enter-echo messages.
By Constraint~(\ref{constraint:D }),

\begin{align*}
\frac{(1-\alpha)^3}{(1+\alpha)^3}-\frac{f}{(1+\alpha)^3 NS_{min}} &\geq \gamma,
\end{align*}

so the value of  $join\_bound$ is at most

\begin{multline*}
\gamma \cdot (1 +\alpha)^3|S|
\leq \left( \frac{(1-\alpha)^3}{(1+\alpha)^3}-\frac{f}{(1+\alpha)^3 NS_{min}}\right) \\ \cdot (1 +\alpha)^3|S|
= (1-\alpha)^3 |S|-f.
\end{multline*}

Thus, by time $t_p^e+2D$,
the condition in line~\ref{line:check if enough enter echoes} of Algorithm~\ref{algo:Common}
holds and node $p$ joins.
\end{proof}

Next, we show that all read and write operations terminate.
Specifically, we show that the number of replies 
 for which an operation
waits is at most the number that it is guaranteed to receive.

Since $enter(q)$ is added to $\mbox{\it Server\_Changes}_p$  whenever $join(q)$
is, for server $q$, we get the following observation.

\begin{observation}
\label{obs:S2}
For every time $t \geq 0$ and  every client $p$ that is active at time $t$,
$\mbox{\it Members}_p^t \subseteq \mbox{\it Present}_p^t$.
\end{observation}

Lemma \ref{lem:present-2D} relates an active node's (client or correct server) current estimate of the
number of servers  present to the number of servers that were present in
the system $2D$ time units earlier.   Lemma \ref{lem:members-2D}
relates an active client's current estimate of the number of servers that are
members to the number of servers that were present in the system $4D$
time units earlier. The lower bounds stated in these lemmas had to take into consideration that Byzantine servers may enter the system and never send a message 
and yet affect the system size. This scenario is impossible in the case of crash failures.  

\begin{lemma}
\label{lem:present-2D}
For every node $p$ that is either a client or a correct server and for every time
$t\geq t_p^j$ at which $p$ is active,
\begin{multline*}
(1-\alpha)^2 \cdot N(\max\{0,t-2D\}) -f \le |\mbox{\it Present}_p^t| \\
\le (1+\alpha)^2 \cdot N(\max\{0,t-2D\}).
\end{multline*}
\end{lemma}

\begin{proof}
The proof is adapted from~\cite{AttiyaCEKW2018} to include $f$ Byzantine servers in the lower bound. 
By Lemma~\ref{lem:knowswhenjoined}, \knows{p}{t}{[0,\max\{0,t-2D\}]}.
Thus {\em Present}$_p^t$ contains all nodes that are present at
time $\max\{0,t-2D\}$,
plus any nodes that  enter in $(\max\{0,t-2D\},t]$ which $p$ has learned
about,
minus any nodes that  leave in $(\max\{0,t-2D\},t]$ which $p$ has learned about.
Then, by Lemma~\ref{lem:size},
\begin{multline*}
|Present_p^t| \leq \\
NS(\max\{0,t-2D\}) +( (1+\alpha)^2-1) \cdot NS(\max\{0,t-2D\}) \\
= (1+\alpha)^2 \cdot NS(\max\{0,t-2D\}).
\end{multline*}
Similarly, by Lemma~\ref{lem:left2} and Constraint~(\ref{constraint:A }),
\begin{multline*}
|Present_p^t| \geq \\
NS(\max\{0,t-2D\}) -(1- (1-\alpha)^2) \cdot NS(\max\{0,t-2D\}) \\
= (1-\alpha)^2 \cdot NS(\max\{0,t-2D\}) .
\end{multline*}
\end{proof}

\sapta{
\begin{lemma}
\label{lem:members-2D}
For every client $p$ and every time $t\geq t_p^j$ at which $p$ is active,
\begin{multline*}
(1-\alpha)^4 \cdot NS(\max\{0,t-4D\}) -f \le
|\mbox{\it Members}_p^t| \\
\le (1+\alpha)^4 \cdot NS(\max\{0,t-4D\}).
\end{multline*}
\end{lemma}
}
\begin{proof}
The proof is adapted from~\cite{AttiyaCEKW2018} to include $f$ Byzantine servers in the lower bound. 
By Lemma~\ref{lem:knowswhenjoined}, \knows{p}{t}{[0,\max\{0,t-2D\}]} and, by
Theorem~\ref{thm:joins}, every node that enters by time $\max\{0,t-4D\}$
joins by time $\max\{0,t-2D\}$ if it is still active.
Thus {\em Members}$_p^t$ contains all nodes that are present at time $\max\{0,t-4D\}$
plus any nodes that enter in $(\max\{0,t-4D\},t]$ which $p$ learns have joined,
minus any nodes that leave in $(\max\{0,t-4D\},t]$ which $p$
learns have left.
Then, by Lemma~\ref{lem:size},
\begin{multline*}
|Members_p^t| \leq \\
NS(\max\{0,t-4D\}) +( (1+\alpha)^4-1) \cdot NS(\max\{0,t-4D\}) \\
= (1+\alpha)^4 \cdot NS(\max\{0,t-4D\}) .
\end{multline*}
Similarly, by Lemma~\ref{lem:left2} and Constraint~(\ref{constraint:A }),
\begin{multline*}
|Members_p^t| \geq \\
NS(\max\{0,t-2D\}) -(1- (1-\alpha)^4) \cdot NS(\max\{0,t-4D\}) \\
= (1-\alpha)^4 \cdot NS(\max\{0,t-4D\}) .
\end{multline*}
\end{proof}

Lemma~\ref{lem:joined-for-D} proves a lower bound on the number of servers that
reply to a client's query or update message. 
\sapta{
\begin{lemma}
\label{lem:joined-for-D}
If a client or correct server $p$ is active at time $t \geq t_p^j$, then the number of correct servers that
are joined by time $t$ and are still active at time $t+D$ is at least
$\left[\frac{(1-\alpha)^3}{(1+\alpha)^2} \right]\cdot|\mbox{\it Present}_p^t| -f$.
\end{lemma}
}

\begin{proof}
By Lemma~\ref{lem:left2} and Constraint~(\ref{constraint:A }),
the maximum number of servers that leave during $(\max\{0,t-2D\},t+D]$
is at most $(1-(1-\alpha)^3) \cdot NS(\max\{0,t-2D\})$.
Thus, there are at least
\begin{multline*}
NS(\max\{0,t-2D\}) - (1-(1-\alpha)^3) \cdot NS(\max\{0,t-2D\}) -f\\
= [(1-\alpha)^3 ] \cdot NS(\max\{0,t-2D\}) -f
\end{multline*}
correct servers that were present at time
$\max\{0,t-2D\}$
and are still active at time $t+D$.
This number is bounded below by
$\left[\frac{(1-\alpha)^3}{(1+\alpha)^2} \right]\cdot |\mbox{\it Present}_p^t|-f $
since, by Lemma~\ref{lem:present-2D},
$NS(\max\{0,t-2D\}) \ge |\mbox{\it Present}_p^t|/(1+\alpha)^2$.
By Theorem~\ref{thm:joins}, all of these servers are joined by time $t$.
\end{proof}

 Lemmas~\ref{lem:present-2D} through~\ref{lem:joined-for-D} 
are used to prove the following theorem. 
\begin{theorem}
\label{thm:ops-live}
Every read or write operation invoked by a client that remains active completes.
\end{theorem}

\begin{proof}
Each operation consists of a read phase and a write phase.
We show that each phase terminates
within $2D$ time, provided the client remains active (does not crash or leave).

Consider a phase of an operation by client $p$ that starts at time $t$.
Every correct server that joins by time $t$ and is still active at time $t+D$
receives $p$'s query or update message
and replies with a  reply message or an ack message
by time $t+D$.
By Lemma~\ref{lem:joined-for-D}, there are at least
$\left[\frac{(1-\alpha)^3}{(1+\alpha)^2}\right]\cdot |\mbox{\it Present}_p^t|-f $
such servers.

From Constraint~(\ref{constraint:E }), Lemma~\ref{lem:present-2D} and Observation~\ref{obs:S2},
\begin{multline*}
\left[\frac{(1-\alpha)^3}{(1+\alpha)^2} \right]\cdot |\mbox{\it Present}_p^t|-f  \geq
           \beta \cdot |\mbox{\it Present}_p^t| \\
  \geq \beta \cdot |\mbox{\it Members}_p^t|  = rw\_bound_p^t.
\end{multline*}
Thus, by time $t+2D$, $p$ receives sufficiently many replies or ack messages
to complete the phase.
\end{proof}

Now we prove atomicity of the \AlgName{} algorithm.
Let $\cal T$ be the set of read operations that complete and
write operations that execute line~\ref{line:bcast update1}
of Algorithm~\ref{algo:Client}.
For any node $p$,
let $ts_p^t = (num_p^t,w\_id_p^t)$ denote the {\em timestamp} of
the latest value  known to node $p$ that is recorded in its \mbox{\it Known\_Writes}$[p]_p$ variable. Note that new timestamps are created by write operations
(on lines~\ref{line:new timestamp}-\ref{line:new timestamp2}  of Algorithm~\ref{algo:Client})
and are sent via enter-echo, update, and update-echo messages.
Initially, $ts_p^0 = (0,\bot)$ for all nodes $p$.

For any operation $o$ in $\cal T$ by client $p$, the {\em
timestamp of its read phase}, $ts^{rp}(o)$, is $ts_p^t$, where $t$
is the end of its read phase (i.e., when the condition on
line~\ref{line:quorum reached} of Algorithm~\ref{algo:Client}
evaluates to true).
The {\em timestamp of its write phase},
$ts^{wp}(o) $, is $ts_p^t$, where $t$ is  the beginning of its write phase
(i.e., when it s-bcasts on line~\ref{line:bcast update1}
of Algorithm~\ref{algo:Client}).
The {\em timestamp of a read operation}  in $\cal T$ is the timestamp of its read phase.
The {\em timestamp of a write operation}  in $\cal T$ is the timestamp of its write phase.

Note that $w\_id$ is equal to $p$ and $num$ is set to one greater than
the largest sequence number occurring in at least $f+1$ replies observed during an operation's read phase.
This implies the next observation:

\begin{observation}\label{obs:unique timestemps}
Each write operation in $\cal T$ has a unique timestamp.
\end{observation}

The next observation follows by a simple induction,
since every timestamp other than $(0,\bot)$ comes from
Lines~\ref{line:new timestamp}-\ref{line:new timestamp2}
of Algorithm~\ref{algo:Client}.

\begin{observation}\label{obs:read values}
Consider any read $op_1$ in $\cal T$.
If the timestamp of a read $op_1$ is $(0,\bot)$, then $op_1$ returns $\bot$.
Otherwise, there is a write $op_2$ in $\cal T$ such that $ts(op_1) = ts(op_2)$
and the value returned by $op_1$ equals the value written by $op_2$.
\end{observation}

If a read operation $op_1$ returns the value written by a write operation $op_2$, then we say that $op_1$ reads from $op_2$. 

Lemmas~\ref{lem:lin1} through~\ref{lem:lin4} show that 
information written in the write phase of an operation propagates properly through the system. It is very important that at every step, the algorithm ensures that outdated information or wrong information sent by Byzantine servers does not corrupt the state of the replicated register. The IsValidMessage() procedure helps mask two types of bad behavior (multiple replies and replies sent after announcing a leave).  
The variables $valid\_val_p$ and $Known\_Writes[]_p$ help mask (bad) replies from Byzantine servers. 
These lemmas are analogous to the Lemmas~\ref{lem:gen1} to~\ref{lem:knowswhenjoined} 
regarding the propagation of information about {\sc Enter}, {\sc Joined}, and {\sc Leave} events. 

\sapta{

\begin{lemma}
\label{lem:lin1}
If $o$ is an operation in $\cal T$ whose write phase w starts at $t_w$,
correct server $p$ is active at time $t \geq t_w+D$, and
$t_p^e \leq t_w$, then $ts_p^t \geq ts^{wp}(o)$.
\end{lemma}
}
\begin{proof}
Since server $p$ is active throughout $[t_w,t_w+D]$, it directly receives the update message
s-bcast by $w$ at time $t_w$.
Hence, from lines~\ref{line:new value?}--\ref{line:update history} 
of Algorithm~\ref{algo:Server},
$ts_p^t \geq ts^{wp}(o)$.
\end{proof}

\sapta{

\begin{lemma}
\label{lem:lin2}
 Suppose a correct server $p \not\in S_0$ receives $(f+1)$ enter-echo messages from correct servers by time $t''$. Let the $f+1$st  enter-echo message from a correct server be received from $q$ that sends it at time $t'$ in reply 
 to an enter message from $p$.
If $o$ is an operation whose write phase $w$ starts at $t_w$,
$p$ is active at time $t \geq \max\{t'',t_w+2D\}$,
and the $f+1$ correct servers that send enter-echo messages are  active throughout $[t_w,t_w + D]$,
then $ts_p^{t} \geq ts^{wp}(o)$.
\end{lemma}
}
\begin{proof}
 By Lemma~\ref{lem:gen0}, there are at least $f+1$ correct joined servers that are active throughout $[t_w, t_w+D]$. 
 Since $q$ is  active throughout
$[t_w,t_w + D]$, it receives the update message from $w$ at some time $\hat{t} \in [t_w,t_w + D]$,
so $ts_q^{\hat{t}} \geq ts^{wp}(o)$.
At time $t'' \leq t$,  $p$ receives the enter-echo sent by  $q$ at time $t'$.  By the above argument, all $f$ earlier enter-echo messages have timestamp $\ge ts^{wp}(o)$. So, the value of the timestamp in $valid\_val_p$ and in $\mbox{\it Known\_Writes} [p]_p$  is set to $\ge ts^{wp}(o)$. So $ts_p^t \geq ts_p^{t''} \geq ts_q^{t'}$.
If $t' \geq \hat{t}$, then $ts_q^{t'} \geq ts_q^{\hat{t}}$, so
$ts_p^{t} \geq ts^{wp}(o)$.
If $\hat{t} > t'$, then $q$ sends an update-echo at time $\hat{t} \leq t_w+D$, and 
$p$ receives it by time $\hat{t}+D  \leq t_w+2D \leq t$. The same argument works for the other $f$ correct, active servers in the system.  Thus the timestamp of the variable $valid\_val_p$ and  $\mbox{\it Known\_Writes} [p]_p$ is either the timestamp of $w$ or of a later write. Thus,
$ts_p^t \geq ts_q^{\hat{t}} \geq ts^{wp}(o)$.
\end{proof}

\sapta{

\begin{lemma}
\label{lem:lin3}
If $o$ is an operation in $\cal T$ whose write phase w starts at $t_w$ and correct
server $p$ is active at time $t \geq \max\{t_p^e+2D,t_w+D\}$,
then $ts_p^t \geq ts^{wp}(o)$.
\end{lemma}
}
\begin{proof}
The proof is by induction on the order in which correct servers enter the system.
Suppose the claim holds for all correct servers that enter before $p$.
If $t_p^e \leq t_w$, which is the case for all $p \in S_0$,
then the claim follows from Lemma~\ref{lem:lin1}.

If $t_w < t_p^e$, then by Lemma~\ref{lem:gen0},
there are at least $f+1$ correct joined servers  that are active throughout
$[\max\{0,t_p^e-2D\},t_p^e+D]$. 
These servers receive an enter message from $p$ 
and send an enter-echo message containing $ts_q^{t'}$ back to $p$. Let $q$ be the server whose enter-echo is the $(f+1)$th enter-echo from a correct joined server to reach $p$.  Let server $q$ receive the enter message from $p$ at some time $t' \in [t_p^e, t_p^e +D]$. The enter-echo  message sent by $q$ is received by $p$ at some time $t'' \leq t'+D \leq t_p^e +2D \leq t$.  So, the value of the timestamp in $valid\_val_p$ and in turn $\mbox{\it Known\_Writes} [p]_p$ for $p$ is set to $ts_p^t\ge ts_p^{t''} \ge ts_q^{t'}$

The first case is when $t_w \geq \max\{0,t_p^e-2D\}$.
Since $t_w + D < t_p^e +D$, it follows that the $f+1$ correct joined servers including $q$ are active throughout $[t_w,t_w+D]$. Furthermore, $t \geq t_p^e+2D \geq \max\{t'', t_w+2D\}$.
Hence, Lemma~\ref{lem:lin2} implies that
$ts_p^{t} \geq ts^{wp}(o)$.

The second case is when $t_w < \max\{0,t_p^e-2D\}$.
Since $t_w \geq 0$, it follows that
$t_p^e-2D > 0$, $t_q^e \leq \max\{0,t_p^e -2D\} = t_p^e -2D$,  and $t_w < t_p^e-2D \leq t'-2D$,
so $t' \geq \max\{t_q^e + 2D, t_w + D\}$.
Note that $q$ is active at time $t'$ and $q$ enters before $p$, so, by the induction hypothesis,
$ts_q^{t'} \geq ts^{wp}(o)$. The above argument is true for all the other $f$ correct joined servers that $p$ hears from.
Hence, $ts_p^t \geq ts^{wp}(o)$.

\end{proof}

\sapta{

\begin{lemma}
\label{lem:lin4}
If $o$ is an operation in $\cal T$ whose write phase starts at $t_w$,
correct server $p \not\in S_0$ joins at time $t_p^j$, and
$p$ is active at time $t \geq \max\{t_p^j,t_w+2D\}$,
then $ts_p^t \geq ts^{wp}(o)$.
\end{lemma}
}

\begin{proof}
The proof is adapted from~\cite{AttiyaCEKW2018} to tolerate $f$ Byzantine servers. 
The proof is by induction on the order in which servers enter the system.
Suppose the claim holds for all servers that join before $p$.
If $t \geq t_p^e +2D$, then the claim follows by Lemma~\ref{lem:lin3}.
So, suppose $t < t_p^e +2D$.
If $t_p^e \leq t_w$, then the claim follows by Lemma~\ref{lem:lin1}.
So, suppose $t_w < t_p^e$.

Before $p$ joins, it receives an enter-echo message from a joined server in reply 
 to its enter message.
Suppose $p$ first receives such an enter-echo message
at time $t''$ and this enter-echo was sent by $q$ at time $t'$. Then $t'' \leq t_p^j \leq t$
and $ts_q^{t'} \leq ts_p^{t''} \leq ts_p^t$.

Now we prove that $p$ receives an enter-echo message from a server $q'$ that is active throughout $[\max\{0,t'-2D\}, t' +D]$.
Let $S$ be the set of servers present at time $\max\{0,t'-2D\}$, so  $|S|= NS(\max\{0,t'-2D\})$. By Lemma~\ref{lem:left2} and Constraint~(\ref{constraint:A }), at most $(1-(1-\alpha)^3)|S|$ servers leave during $(\max\{0,t'-2D\},t'+D]$. Since $t'' \leq t'+D$, it follows that $|\mbox{\it Present}_p^{t''}| \geq |S| - (1-(1-\alpha)^3)|S| = (1-\alpha)^3|S|$.
Hence, from lines~\ref{line:calculate join bound}
and~\ref{line:check if enough enter echoes}
of Algorithm~\ref{algo:Common}, $p$ waits until it has received at least $join\_bound = \gamma \cdot |\mbox{\it Present}_p^{t''}|  \geq \gamma \cdot (1-\alpha)^3|S| $ enter-echo messages before joining.

\sapta{Hence, the number of enter-echo messages $p$ receives before joining from servers that were active throughout $[\max\{0,t'-2D\},t'+D]$ is
$join\_bound$ minus the total number of server enters, leaves and faults (as Byzantine servers may not reply at all),
which is at least

\begin{multline}
\gamma \cdot (1-\alpha)^3 |S| - 
    [((1+\alpha)^3 -1)|S| + (1-(1-\alpha)^3)|S|  +f] \\
= [(1+\gamma) (1-\alpha)^3 - (1+\alpha)^3]|S| -f \\
\geq  [(1+\gamma) (1-\alpha)^3 - (1+\alpha)^3]NS_{min} -f \label{Lemma9equation}
\end{multline}

Rearranging Constraint~(\ref{constraint:C }), we get
\[
[(1+\gamma)(1- \alpha)^3 -(1+\alpha)^3]NS_{min}-f \geq f+1 ,
\]
so expression~(\ref{Lemma9equation}) is at least $f+1$.
Hence $p$ receives an enter-echo message at some time $T'' \leq t_p^j$
from a correct server $q'$ that is active throughout
\[
[\max\{0,t'-2D\},t'+D] \supseteq [\max\{0,t'-2D\},t-D] .
\]}
Let $T'$ be the time that $q'$ sent its enter-echo message in reply 
 to the enter message from $p$.
Then $ts_{q'}^{T'} \leq ts_p^{T''} \leq ts_p^t$.

Note that $t_w < t_p^e \leq t'$ so $t_w + D \leq t'+D$.
If $t_w \geq \max\{0,t'-2D\}$, then $q'$ is active throughout $[t_w,t_w+D]$.
Since $t \geq \max\{T'',t_w+2D\}$, it follows by Lemma~\ref{lem:lin2} that $ts_p^t \geq ts^{wp}(o)$.
So, suppose $t_w < \max\{0,t'-2D\}$.

Since $t_w \geq 0$, it follows that $t' > t_w + 2D$.
If $q \in S_0$, then $t_q^e = 0 \leq t_w$, so, by Lemma~\ref{lem:lin1},
$ts_q^{t'} \geq ts^{wp}(o)$.
If $q \not\in S_0$, then, by the induction hypothesis,
$ts_q^{t'} \geq ts^{wp}(o)$, since $q$ joins at time $t_q^j < t_p^j \leq t'$.
Thus, in both cases, $ts_p^{t} \geq ts^{wp}(o)$.
\end{proof}

Lemmas~\ref{lem:lin1} through~\ref{lem:lin4} are used to prove Lemma~\ref{lem:lin5}, which is the key lemma for proving atomicity of \AlgName{}.
It shows that for two non-overlapping operations,
the timestamp of the read phase of the latter operation is at least
as \sapta{large} as the timestamp of the write phase of the former.

\sapta{
\begin{lemma}\label{lem:lin5}
For any two operations $op_1$ and $op_2$ in $\cal T$, if $op_1$
finishes before $op_2$ starts, then $ts^{wp}(op_1) \leq ts^{rp}(op_2)$.
\end{lemma}
}

\begin{proof}
Let $p_1$ be the client that invokes $op_1$, let $w$ denote the write phase of $op_1$,
let $t_w$ be the start time of $w$, and let $\tau_w = ts^{wp}(op_1) = ts^{t_w}_{p_1}$.
Similarly, let $p_2$ be the client that invokes $op_2$, let $r$ denote the read phase of $op_2$,
let $t_r$ be the start time of $r$, and let $\tau_r = ts^{rp}(op_2) = ts^{t_r}_{p_2}$.

Let $Q_w$ be the set of servers that $p_1$ hears from during $w$
(i.e., that sent messages causing $p_1$ to increment $rw\_counter$
on line~\ref{line:inc heard from ack} of Algorithm~\ref{algo:Client})
and $Q_r$ be the set of servers that $p_2$ hears from during $r$
(i.e., that sent messages causing $p_2$ to increment $rw\_counter$
on line~\ref{line:inc heard from resp} of Algorithm~\ref{algo:Client}).
Let $P_w = |Present_{p_1}^{t_w}| $ and $M_w = |Members_{p_1}^{t_w}| $
be the sizes of the {\em Present} and {\em Members} sets belonging
to $p_1$ at time $t_w$, and $P_r = |Present_{p_2}^{t_r}|$
and $M_r  = |Members_{p_2}^{t_r}| $ be the sizes of the {\em Present}
and {\em Members} sets belonging to $p_2$ at time $t_r$.

\medskip
\noindent
{\bf Case I:} $t_r > t_w + 2D$.
We start by showing that there exists $f+1$ correct servers  in $Q_r$ such that
$t_q^j \leq t_r-2D$. 

Each server $q \in Q_r$ receives and responds to
$r$'s query, so $q$ is joined by time $t_r+D$. By Theorem~\ref{thm:joins},
the number of servers that can join during $(t_r - 2D, t_r + D]$ is at
most the number of servers that can enter in $(\max\{0,t_r - 4D\}, t_r
+ D]$. By Lemma~\ref{lem:size}, the number of servers that can enter
during $(\max\{0,t_r - 4D\}, t_r + D]$ is at most $((1 + \alpha)^5
- 1) \cdot NS(\max\{0,t_r - 4D\})$. By
Lemma~\ref{lem:members-2D}, $(1-\alpha)^4 NS(\max\{0,t_r - 4D\}) -f \le
M_r$. 

From the code and Constraint~(\ref{constraint:F }), it follows that 
\begin{align*}
|Q_r| \ge \beta M_r& >\left[\frac{ (1 + \alpha)^5 - 1 +2f/NS_{min} }{(1 - \alpha)^4 -f/NS_{min}} \right]\cdot M_r
\\ &\geq \left[\frac{ (1 + \alpha)^5 - 1 +2f/NS_{min} }{(1 - \alpha)^4 -f/NS_{min}} \right]\\&\cdot ((1 - \alpha)^4  NS(\max\{0,t_r - 4D\}) -f)
\\ &\geq [(1+\alpha)^5-1]\cdot NS(\max\{0,t_r - 4D\}) +2f 
\end{align*}
which is $2f+1$ more than the maximum
 number of servers that can enter in $(\max\{0,t_r - 4D\}, t_r+ D)$. At most $f$ of these can be Byzantine. Thus, at least $f+1$ correct servers in $ Q_r$ join by time $t_r-2D$.

Suppose correct server $q\in Q_r$ receives $r$'s query message at time $t' \geq t_r \geq t_w
+ 2D$.  If $q \in S_0$, then $t_q^j = 0 \leq t_w$, so, by
Lemma~\ref{lem:lin1}, $ts_q^{t'} \ge ts^{wp}(op_1) =
\tau_w$.  Otherwise, $q \not\in S_0$, so $0 < t_q^j \leq t_r - 2D <
t'$.  Since $t_w+2D < t_r \leq t'$, Lemma~\ref{lem:lin4} implies that
$ts_q^{t'} \ge ts^{wp}(op_1) = \tau_w$.  In either case, $q$ responds
to $r$'s query message with a timestamp at least as large as $\tau_w$
and, hence, $\tau_r \geq \tau_w$.

\medskip
\noindent
{\bf Case II:} $t_r \le t_w + 2D$.
Let
$J = \{ p~|~t_p^j < t_r$ and $p$ is an active server at time $t_r\} \cup \{p~|~ t_r \leq t_p^j \leq t_r+D\}$, which contains the set of all servers that reply to $r$'s query. %
By Theorem~\ref{thm:joins},   all correct servers that are present at time $t_r-2D$ join by time $t_r$ if they remain active. Therefore all servers in $J$ are either active at time $\max\{0,t_r-2D\}$ 
or enter during $(\max\{0,t_r-2D\},t_r+D]$.
By Lemma~\ref{lem:size}, $|J| \leq (1+\alpha)^3NS(\max\{0,t_r-2D\})$.

Let $K$ be the set of all servers that are present at time $\max\{0,t_r-2D\}$ and do not leave 
 during $(\max\{0,t_r-2D\},t_r+D]$.
Note that $K$ contains all the servers in $Q_w$ that do not leave 
 during $[t_w, t_r + D] \subseteq [\max\{0,t_r -2D\},t_r+D]$.
By Lemma~\ref{lem:left2} and Constraint~(\ref{constraint:A }), at most $(1-(1-\alpha)^3)NS(\max\{0,t_r-2D\})$ servers leave during $[\max\{0,t_r -2D\},t_r+D]$. 

From the code, $|Q_r| \geq \beta M_r $ and, by Lemma~\ref{lem:members-2D}, $M_r \geq (1-\alpha)^4  NS(\max\{0,t_r-4D\})-f$. So,
\begin{align*}
|Q_r| \geq \beta \left[ (1-\alpha)^4  NS(\max\{0,t_r-4D\})-f\right].
\end{align*}
Similarly,
$$|Q_w| \geq \beta M_w  \geq  \beta \left[(1-\alpha)^4  NS(\max\{0,t_w-4D\})-f\right]$$
Therefore, the size of $K$ is at least
\begin{align*}
|K| &\ge |Q_w| -(1-(1-\alpha)^3 ) NS( \max\{0,t_r-2D\})-f\\
&\ge \beta \left[(1-\alpha)^4 NS(\max\{0,t_w-4D\}) -f \right] \\
&- (1-(1-\alpha)^3) NS( \max\{0,t_r-2D\}) - f. \label{inequality:k}
\end{align*}

Since $t_r - t_w \le 2D$, it follows that
$\max\{0,t_r - 4D\} - \max\{0,t_w - 4D\} \le 2D$.
By Lemma~\ref{lem:size}, $NS(\max\{0,t_r - 4D\}) \le (1+\alpha)^2 \cdot NS(\max\{0,t_w - 4D\})$. Thus we can replace $NS(\max\{0,t_w - 4D\})$ in the above expression with $(1+\alpha)^{-2} \cdot NS(\max\{0, t_r - 4D\})$ and get the following expression for $|Q_r| +|K|$:
\begin{align*}
&|Q_r|+ |K| \\
&\geq
   \beta [(1-\alpha)^4 NS(\max\{0,t_r-4D\}) -f]\\
    &+\beta [(1-\alpha)^4 (1+\alpha)^{-2} NS(\max\{0,t_r-4D\})-f] \\
    & \hspace{1in}- (1-(1-\alpha)^3)NS(\max\{0,t_r-2D\}) -f\\
&= \beta [(1-\alpha)^4 (1+\alpha)^{-2} (2+2\alpha+\alpha^2) NS(\max\{0,t_r-4D\})-2f] \\
    & \hspace{1in}-(1 - (1-\alpha)^3  )NS(\max\{0,t_r-2D\}) -f.
\end{align*}
By Lemma~\ref{lem:size},
$$NS(\max\{0,t_r-4D\}) \geq (1-\alpha)^{-2} NS(\max\{0,t_r-2D\}) . $$
Thus,
\begin{align*}
|Q_r|+ |K|& \geq \beta [(1-\alpha)^2 (1+\alpha)^{-2} \\&(2+2\alpha+\alpha^2)NS(\max\{0,t_r-2D\}) -2f]\\
& -(1 - (1-\alpha)^3 )NS(\max\{0,t_r-2D\}) -f
\end{align*}
By Constraint~(\ref{constraint:G }),
$$\beta  > \frac{(1 + \alpha)^3 - (1 - \alpha)^3 +1+ (1+3f)/NS_{min}}{\left[(2 + 2\alpha + \alpha^2)(1 - \alpha)^2(1 + \alpha)^{-2}\right] -2f/NS_{min}}  $$
Let $N_{tr2D} = NS(\max\{0,t_r-2D\})$
So,
\begin{align*}
&|Q_r| + |K|\\
&\ge\left( \frac{\left((1 + \alpha)^3 - (1 - \alpha)^3 +1+\frac{ (1+3f)}{NS_{min}}\right) }{\left((2 + 2\alpha + \alpha^2)(1 - \alpha)^2(1 + \alpha)^{-2}\right) -2f/NS_{min}} \right) \\ &\cdot \left((2+2\alpha+\alpha^2)(1-\alpha)^2 (1+\alpha)^{-2}   -\frac{2f}{N_{tr2D}}\right)N_{tr2D}\\&\hspace{1in} -(1-(1 - \alpha)^3)] N_{tr2D} -f  
 \end{align*}
 Since, $N_{tr2D} \ge NS_{min}$, we get
 \begin{align*}
&|Q_r| + |K|\\
&\ge\left((1 + \alpha)^3 - (1 - \alpha)^3 +1+ (1+3f)/NS_{min}\right)N_{tr2D}\\ &-f  -(1-(1 - \alpha)^3)] N_{tr2D} -f \\
&\ge (1 + \alpha)^3 N_{tr2D} + 2f+1 \ge |J| +2f+1.
\end{align*}
This implies that the intersection of $K$ and $Q_r$ has at least $2f+1$ servers. 
For each server $p$ in the intersection, $ts_p \ge \tau_w$
when $p$ sends its reply 
 to $r$. Since at most $f$ servers can be Byzantine, there are at least $f+1$ correct servers that reply with $ts_p \ge \tau_w$. Thus the timestamp of $valid\_val_p$ on Line number~\ref{line:valid val} is $\ge \tau_w$, thus, $\tau_w\le \tau_r$.
\end{proof}

The proof of Theorem~\ref{thm:atomicity} uses Lemma~\ref{lem:lin5} to show that
the timestamps of two non-overlapping operations respect real time ordering
and completes the proof of atomicity.
\begin{theorem}
 \label{thm:atomicity}
\AlgName{} ensures atomicity.
\end{theorem}

\begin{proof}
This proof is taken from from Theorem 23  in~\cite{AttiyaCEKW2018}. 
We show that, for every execution, there is a total order on the set
of all completed read operations and all write operations that execute
Line~\ref{line:bcast update1} of Algorithm~\ref{algo:Client} such
that every read returns the value of the latest preceding write
(or the initial value if there is no preceding write) and, if
an operation $op_1$ finishes before another operation $op_2$ begins,
then $op_1$ is ordered before $op_2$.

 We first order the write operations in order of their (unique)
timestamps.
 Then, we go over all reads in the ordering of the start times,
and place a read with timestamp $(0,\bot)$ at the beginning of the total order.
Place every other read  after the write operation it reads from ,
and after all the previous reads that read from this write operation.
By the Observation~\ref{obs:read values}, every read in the total
order returns the value of the latest preceding write
(or $\bot$ if there is no preceding write).

 We show that the total order respects the real-time
order of non-overlapping operations in the execution.
Let $op_1$ and $op_2$ be two operations in $\cal T$ such that
$op_1$ finishes before $op_2$ starts.
By the definition of timestamps,
$ts(op_1) \le ts^{wp}(op_1)$ and $ts^{rp}(op_2) \le ts(op_2)$.
By Lemma~\ref{lem:lin5}, $ts^{wp}(op_1) \le ts^{rp}(op_2)$.
Therefore, if $op_2$ is a read, then
\begin{equation}\label{lem:lin6a}
ts(op_1) \leq ts(op_2)
\end{equation}
If $op_2$ is a write, then $ts^{wp}(op_2) = ts^{rp}(op_2) + 1$,
and
\begin{equation}\label{lem:lin6b}
ts(op_1) < ts(op_2)
\end{equation}

We consider the following cases:
\begin{itemize}

\item  Suppose $op_1$ and $op_2$ are both writes.
  By~(\ref{lem:lin6b}), $ts(op_1) < ts(op_2)$ and
  thus the construction orders $op_1$ before $op_2$.

\item Suppose $op_1$ is a write and $op_2$ is a read.  By~(\ref{lem:lin6a})
  and the construction, $op_2$ is placed
  after the write $op_3$ that $op_2$ reads from.  If $ts(op_1) =
  ts(op_2)$ then $op_1 = op_3$ and $op_2$ is placed after $op_1$. If
  $ts(op_1) < ts(op_2)$ then $op_3$ is placed after $op_1$ as
  $ts(op_1) < ts(op_3)$ and thus $op_2$ is placed after $op_1$ in the
  total order.

\item Suppose $op_1$ is a read and $op_2$ is a write.   By~\ref{lem:lin6b},
  $ts(op_1) < ts(op_2)$.  Now, either
  $op_2$ is the first write in the execution and $op_1$'s timestamp is
  $(0,\bot)$ or there exists another write $op_3$ that $op_1$ reads
  from.  If $op_1$'s timestamp is $(0,\bot)$ then the construction
  orders $op_1$ before $op_2$. Otherwise, the construction orders
  $op_3$ before $op_2$.  Since $op_1$ is ordered after $op_3$ but
  before any subsequent write, $op_1$ precedes $op_2$ in the total
  order.

\item Finally, suppose that $op_1$ and $op_2$ are both reads.  By~\ref{lem:lin6a},
  $ts(op_1) \le ts(op_2)$. If $op_1$ and
  $op_2$ have the same timestamp, then they are placed after the same
  write (or before the first write) and the construction orders them
  based on their starting times. Since $op_1$ completes before $op_2$
  starts, the construction places $op_1$ before $op_2$.  If $op_2$ has
  a timestamp greater than that of $op_1$, then $ts(op_2)$ cannot be
  $(0,\perp)$ and so there is a write operation $op_3$ whose timestamp
  is greater than that of $op_1$ and equal to that of $op_2$. The
  construction places $op_1$ before $op_3$ and $op_2$ after $op_3$.
\end{itemize}
Thus, {\sc \AlgName{}} ensures atomicity.
\end{proof}



\section{Discussion}
\label{section:discussion}
Our paper provides an algorithm that emulates a Byzantine-tolerant atomic register in a dynamic system (i) where consensus is impossible to solve, (ii) that never stops changing and (iii) has no upper bound on the system size. We also provide an impossibility  proof that in our model, a uniform algorithm cannot be implemented. 

There are several directions for future work. The values of $\alpha$, $f$ and $NS_{min}$ that satisfy our algorithm are quite restrictive. It will be nice to see if such restrictions are necessary or if they can be improved either with a better algorithm or a tighter correctness analysis.

Currently our  model  tolerates the most severe end of the fault severity spectrum and paper~\cite{AttiyaCEKW2018} considered the most benign end of the fault severity spectrum. 
 In future, we would like to 
   to see if our impossibility result extends to the other failure models?

Our current way of restricting churn relies on the unknown upper bound $D$ on message delay. 
\sapta{An alternative may be to allow the upper bound on the message delays 
to be unbounded and to
   define the churn rate with respect to messages in transit.}
For example, at all times, the number of servers that can enter/leave the system when any message is in transit at most $\alpha$ times the system size when the message was sent. It may be possible to prove that the
two models are indeed equivalent, or to show that the algorithms still work, or can be modified to work, in
the new model.

\bibliography{references}

\begin{thebibliography}{10}

\bibitem{AbrahamCKM2007}
I.~Abraham, G.~Chockler, I.~Keidar, and D.~Malkhi.
\newblock Wait-free regular storage from byzantine components.
\newblock {\em Information Processing Letters}, 101(2):60 -- 65, 2007.

\bibitem{AguileraKMS2011}
M.~K. Aguilera, I.~Keidar, D.~Malkhi, and A.~Shraer.
\newblock Dynamic atomic storage without consensus.
\newblock {\em J. ACM}, 58(2):7, 2011.

\bibitem{AlchieriBGF18}
E.~A.~P. Alchieri, A.~Bessani, F.~Greve, and J.~da~Silva~Fraga.
\newblock Knowledge connectivity requirements for solving byzantine consensus
  with unknown participants.
\newblock {\em {IEEE} Trans. Dependable Sec. Comput.}, 15(2):246--259, 2018.

\bibitem{AlvisiPMRW00}
L.~Alvisi, E.~T. Pierce, D.~Malkhi, M.~K. Reiter, and R.~N. Wright.
\newblock Dynamic byzantine quorum systems.
\newblock In {\em 2000 International Conference on Dependable Systems and
  Networks {(DSN} 2000) (formerly {FTCS-30} and DCCA-8), 25-28 June 2000, New
  York, NY, {USA}}, pages 283--292, 2000.

\bibitem{AngluinAFJ08}
D.~Angluin, J.~Aspnes, M.~J. Fischer, and H.~Jiang.
\newblock Self-stabilizing population protocols.
\newblock {\em {TAAS}}, 3(4):13:1--13:28, 2008.

\bibitem{AttiyaBD1995}
H.~Attiya, A.~Bar-Noy, and D.~Dolev.
\newblock Sharing memory robustly in message-passing systems.
\newblock {\em J. ACM}, 42(1):124--142, Jan. 1995.

\bibitem{AttiyaB2003}
H.~Attiya and A.~Bar-Or.
\newblock Sharing memory with semi-byzantine clients and faulty storage
  servers.
\newblock In {\em 22nd International Symposium on Reliable Distributed Systems,
  2003. Proceedings.}, pages 371--378, Oct 2003.

\bibitem{AttiyaCEKW2015}
H.~Attiya, H.~C. Chung, F.~Ellen, S.~Kumar, and J.~L. Welch.
\newblock Simulating a shared register in an asynchronous system that never
  stops changing.
\newblock In {\em Proceedings of 29th International Symposium on Distributed
  Computing}, pages 75--91, 2015.

\bibitem{AttiyaCEKW2018}
H.~{Attiya}, H.~C. {Chung}, F.~{Ellen}, S.~{Kumar}, and J.~L. {Welch}.
\newblock Emulating a shared register in a system that never stops changing.
\newblock {\em IEEE Transactions on Parallel and Distributed Systems},
  30(3):544--559, March 2019.

\bibitem{BaldoniBKR2009}
R.~Baldoni, S.~Bonomi, A.-M. Kermarrec, and M.~Raynal.
\newblock Implementing a register in a dynamic distributed system.
\newblock In {\em IEEE International Conference on Distributed Computing
  Systems}, pages 639--647, 2009.

\bibitem{BaldoniBN13}
R.~Baldoni, S.~Bonomi, and A.~S. Nezhad.
\newblock A protocol for implementing byzantine storage in churn-prone
  distributed systems.
\newblock {\em Theor. Comput. Sci.}, 512:28--40, 2013.

\bibitem{BaldoniBR12}
R.~Baldoni, S.~Bonomi, and M.~Raynal.
\newblock Implementing a regular register in an eventually synchronous
  distributed system prone to continuous churn.
\newblock {\em {IEEE} Trans. Parallel Distrib. Syst.}, 23(1):102--109, 2012.

\bibitem{BonomiPPT16}
S.~Bonomi, A.~D. Pozzo, M.~Potop{-}Butucaru, and S.~Tixeuil.
\newblock Optimal mobile byzantine fault tolerant distributed storage: Extended
  abstract.
\newblock In {\em Proceedings of the 2016 {ACM} Symposium on Principles of
  Distributed Computing, {PODC} 2016, Chicago, IL, USA, July 25-28, 2016},
  pages 269--278, 2016.

\bibitem{CastroL99}
M.~Castro and B.~Liskov.
\newblock Practical byzantine fault tolerance.
\newblock In {\em Proceedings of the Third {USENIX} Symposium on Operating
  Systems Design and Implementation (OSDI), New Orleans, Louisiana, USA,
  February 22-25, 1999}, pages 173--186, 1999.

\bibitem{CastroL02}
M.~Castro and B.~Liskov.
\newblock Practical byzantine fault tolerance and proactive recovery.
\newblock {\em {ACM} Trans. Comput. Syst.}, 20(4):398--461, 2002.

\bibitem{DantasBFC2007}
W.~S. Dantas, A.~N. Bessani, J.~d.~S.~Fraga, and M.~Correia.
\newblock Evaluating byzantine quorum systems.
\newblock In {\em 2007 26th IEEE International Symposium on Reliable
  Distributed Systems (SRDS 2007)}, pages 253--264, Oct 2007.

\bibitem{FLP}
M.~J. Fischer, N.~A. Lynch, and M.~S. Paterson.
\newblock Impossibility of distributed consensus with one faulty process.
\newblock {\em J. ACM}, 32(2):374--382, 1985.

\bibitem{HerlihyW1990}
M.~P. Herlihy and J.~M. Wing.
\newblock Linearizability: {A} correctness condition for concurrent objects.
\newblock {\em Trans. on Prog. Lang. Sys.}, 12(3):463--492, July 1990.

\bibitem{KoHG2008}
S.~Y. Ko, I.~Hoque, and I.~Gupta.
\newblock Using tractable and realistic churn models to analyze quiescence
  behavior of distributed protocols.
\newblock In {\em IEEE Symposium on Reliable Distributed Systems}, pages
  259--268, 2008.

\bibitem{Lamport86}
L.~Lamport.
\newblock On interprocess communication. part {I:} basic formalism.
\newblock {\em Distributed Computing}, 1(2):77--85, 1986.

\bibitem{Lamport86a}
L.~Lamport.
\newblock On interprocess communication. part {II:} algorithms.
\newblock {\em Distributed Computing}, 1(2):86--101, 1986.

\bibitem{LamportSP82}
L.~Lamport, R.~E. Shostak, and M.~C. Pease.
\newblock The byzantine generals problem.
\newblock {\em {ACM} Trans. Program. Lang. Syst.}, 4(3):382--401, 1982.

\bibitem{LindellLR06}
Y.~Lindell, A.~Lysyanskaya, and T.~Rabin.
\newblock On the composition of authenticated byzantine agreement.
\newblock {\em J. {ACM}}, 53(6):881--917, 2006.

\bibitem{LynchS1997}
N.~A. Lynch and A.~A. Shvartsman.
\newblock Robust emulation of shared memory using dynamic quorum-acknowledged
  broadcasts.
\newblock In {\em Proceedings of the 27th International Symposium on
  Fault-Tolerant Computing}, pages 272--281, 1997.

\bibitem{LynchS2002}
N.~A. Lynch and A.~A. Shvartsman.
\newblock Rambo: {A Reconfigurable Atomic Memory Service for Dynamic Networks}.
\newblock In {\em Proceedings of the 16th International Conference on
  Distributed Computing}, pages 173--190, 2002.

\bibitem{MalkhiR98}
D.~Malkhi and M.~K. Reiter.
\newblock Byzantine quorum systems.
\newblock {\em Distributed Computing}, 11(4):203--213, 1998.

\bibitem{MalkhiRW00}
D.~Malkhi, M.~K. Reiter, and A.~Wool.
\newblock The load and availability of byzantine quorum systems.
\newblock {\em {SIAM} J. Comput.}, 29(6):1889--1906, 2000.

\bibitem{MartinAD02}
J.~Martin, L.~Alvisi, and M.~Dahlin.
\newblock Small byzantine quorum systems.
\newblock In {\em 2002 International Conference on Dependable Systems and
  Networks {(DSN} 2002), 23-26 June 2002, Bethesda, MD, USA, Proceedings},
  pages 374--388, 2002.

\bibitem{MartinAD2002}
J.-P. Martin, L.~Alvisi, and M.~Dahlin.
\newblock Minimal byzantine storage.
\newblock In {\em Proceedings of the 16th International Conference on
  Distributed Computing}, pages 311--325, 2002.

\bibitem{MurphyL00}
B.~Murphy and B.~Levidow.
\newblock Windows 2000 dependability.
\newblock {\em IEEE International Conference on Dependable Systems and
  Networks.}, 2000.

\bibitem{RivestSA1978}
R.~L. Rivest, A.~Shamir, and L.~Adleman.
\newblock A method for obtaining digital signatures and public-key
  cryptosystems.
\newblock {\em Commun. ACM}, 21(2):120--126, Feb. 1978.

\bibitem{Tsudik1992}
G.~Tsudik.
\newblock Message authentication with one-way hash functions.
\newblock {\em SIGCOMM Comput. Commun. Rev.}, 22(5):29--38, Oct. 1992.

\bibitem{Vukolic2012}
M.~Vukolic.
\newblock {\em Quorum Systems: With Applications to Storage and Consensus}.
\newblock Synthesis Lectures on Distributed Computing Theory. Morgan \&
  Claypool Publishers, 2012.

\end{thebibliography}
\bibliographystyle{abbrv}
\newpage
\end{document}